\documentclass[12pt]{article}%
\usepackage[cm]{fullpage}%
\usepackage{mleftright}%
\usepackage{amsmath}%
\usepackage{amssymb}
\usepackage{mathtools}
\usepackage{booktabs}

\newcommand{\deq}{\coloneqq}
\usepackage{enumerate}
\usepackage[boxruled,algosection,noend]{algorithm2e}
\usepackage{algorithmic}
\usepackage{graphicx}
\usepackage{xcolor}
\usepackage[inline]{enumitem}
\usepackage{scalerel}
\usepackage{needspace}
\usepackage{caption}%

\usepackage{hyperref}%
\hypersetup{%
   breaklinks,%
   ocgcolorlinks, colorlinks=true,%
   urlcolor=[rgb]{0.25,0.0,0.0},%
   linkcolor=[rgb]{0.5,0.0,0.0},%
   citecolor=[rgb]{0,0.2,0.445},%
   filecolor=[rgb]{0,0,0.4}, anchorcolor=[rgb]={0.0,0.1,0.2}%
}

\usepackage{titlesec}%
\titlelabel{\thetitle. }%

\usepackage{changes}%
\definechangesauthor[name=siva, color=red]{s}
\definechangesauthor[name=makrand, color=orange]{m}
\definechangesauthor[name=cyrus, color=pink]{c}
\definechangesauthor[name=paul, color=green]{p}%

\usepackage[amsmath,thmmarks]{ntheorem}%
\theoremseparator{.}%

\theoremstyle{plain}%
\newtheorem{theorem}{Theorem}[section]
\newtheorem{claim}[theorem]{Claim}
\newtheorem{lemma}[theorem]{Lemma}

\newtheorem{corollary}[theorem]{Corollary}

\theoremstyle{remark}%
\theoremheaderfont{\sf}%
\theorembodyfont{\upshape}%

\newtheorem*{defn:unnumbered}[FakeCounter]{Definition}

\newtheorem*{remark:unnumbered}[FakeCounter]{Remark}%
\newtheorem{remark}[theorem]{Remark}%
\newcommand{\myqedsymbol}{\rule{2mm}{2mm}}

\theoremheaderfont{\em}%
\theorembodyfont{\upshape}%
\theoremstyle{nonumberplain}%
\theoremseparator{}%
\theoremsymbol{\myqedsymbol}%
\newtheorem{proof}{Proof:}%

\numberwithin{equation}{section}
\numberwithin{figure}{section}
\numberwithin{table}{section}

\newcommand{\etal}{\textit{et~al.}\xspace}

\renewcommand{\th}{\ensuremath{\hphantom{}^{\,\mathrm{th}}}\xspace}

\newcommand{\HLinkShort}[2]{\hyperref[#2]{#1\ref*{#2}}}
\newcommand{\HLink}[2]{\hyperref[#2]{#1~\ref*{#2}}}
\newcommand{\HLinkPage}[2]{\hyperref[#2]{#1~\ref*{#2}%
      $_\text{p\pageref{#2}}$}}
\newcommand{\HLinkPageOnly}[1]{\hyperref[#1]{Page~\refpage*{#1}%
      $_\text{p\pageref{#1}}$}}

\newcommand{\HLinkSuffix}[3]{\hyperref[#2]{#1\ref*{#2}{#3}}}
\newcommand{\HLinkPageSuffix}[3]{\hyperref[#2]{#1\ref*{#2}%
      #3$_\text{p\pageref{#2}}$}}

\newcommand{\figlab}[1]{\label{fig:#1}}
\newcommand{\figref}[1]{\HLink{Figure}{fig:#1}}

\newcommand{\seclab}[1]{\label{sec:#1}}
\newcommand{\secref}[1]{\HLink{Section}{sec:#1}}

\newcommand{\clmlab}[1]{\label{claim:#1}}
\newcommand{\clmref}[1]{\HLink{Claim}{claim:#1}}

\newcommand{\itemlab}[1]{\label{item:#1}}
\newcommand{\itemref}[1]{\HLinkSuffix{}{item:#1}{}}

\newcommand{\remlab}[1]{\label{rem:#1}}
\newcommand{\remref}[1]{\HLink{Remark}{rem:#1}}%

\newcommand{\lemlab}[1]{\label{lemma:#1}}
\newcommand{\lemref}[1]{\HLink{Lemma}{lemma:#1}}%

\newcommand{\tablab}[1]{\label{table:#1}}%
\newcommand{\tabref}[1]{\HLink{Table}{table:#1}}%

\newcommand{\alglab}[1]{\label{Algorithm:#1}}%
\newcommand{\algref}[1]{\HLink{Algorithm}{Algorithm:#1}}%

\newcommand{\thmlab}[1]{{\label{theo:#1}}}
\newcommand{\thmref}[1]{\HLink{Theorem}{theo:#1}}

\providecommand{\eqlab}[1]{}%
\renewcommand{\eqlab}[1]{\label{equation:#1}}

\newcommand{\Eqref}[1]{\HLinkSuffix{Eq.~(}{equation:#1}{)}}

\providecommand{\Mh}[1]{{#1}}%

\IfFileExists{.latex_printer_friendly}{\def\GenPrinterVer{1}}{}%

\ifx\GenPrinterVer\undefined
   \IfFileExists{.latex_color}{\def\GenColorMath{1}}{}
\else
\fi

\ifx\GenColorMath\undefined
\else
\renewcommand{\Mh}[1]{{\textcolor{red}{#1}}}%
\fi
\providecommand{\si}[1]{#1}

\newcommand{\poly}{\mathrm{poly}}
\newcommand{\calU}{\mathcal{U}}

\newcommand{\pbrcx}[1]{\left[ {#1} \right]}%

\newcommand{\cardin}[1]{\left| {#1} \right|}%
\newcommand{\pth}[2][\!]{\mleft({#2}\mright)}%
\newcommand{\brc}[1]{\left\{ {#1} \right\}}
\newcommand{\eps}{{\varepsilon}}%
\newcommand{\epsA}{\Mh{\xi}}%

\newcommand{\tm}{\Mh{\widetilde{e}}}%

\newcommand{\idx}{\Mh{\alpha}}%

\newcommand{\ExChar}{\mathbb{E}}%
\newcommand{\Ex}[1]{\ExChar\mleft[ #1 \mright]}
\newcommand{\Var}[1]{\mathop{\mathbb{V}}\mleft[ #1 \mright]}%

\newcommand{\ExCond}[2]{\ExChar\mleft[ #1 \;\middle\vert\; #2 \mright]}

\newcommand{\ProbCond}[2]{\ProbChar\mleft[%
       #1 \;\middle\vert\; #2 \mright]}
\newcommand{\ProbChar}{\mathbb{P}}
\newcommand{\Prob}[1]{\ProbChar\!\pbrcx{#1}}

\newcommand{\e}{\varepsilon}
\newcommand{\ce}{c_{\mathsf{ce}}}

\newcommand{\wt}{\widetilde}%
\newcommand{\wtE}{\Mh{\widetilde{e}}}%
\newcommand{\sq}{\subseteq}

\newcommand{\estm}{\Mh{\widetilde{\nEdges}}}
\newcommand{\est}[1]{\widetilde{#1}}

\newcommand{\polylog}{\mathrm{polylog}}

\usepackage{stmaryrd}%
\newcommand{\IntRange}[1]{\mleft\llbracket #1 \mright\rrbracket}
\newcommand{\IRX}[1]{\IntRange{#1}}%
\newcommand{\IRY}[2]{\left\llbracket #1:#2 \right\rrbracket}

\newcommand{\ColX}[1]{\Mh{n}_{#1}}
\newcommand{\ColLTY}[2]{\Mh{C}_{\!<#1}^{\,#2}}
\newcommand{\ColGTY}[2]{\Mh{C}_{>#1}^{#2}}

\newcommand{\Term}[1]{\textsf{#1}}

\newcommand{\IS}{\Term{IS}\xspace}
\newcommand{\BIS}{\Term{BIS}\xspace}
\newcommand{\hrefb}[3][black]{\href{#2}{\color{#1}{#3}}}%

\definecolor{blue25}{rgb}{0, 0, 11}
\newcommand{\emphic}[2]{%
   \textcolor{blue25}{%
      \textbf{\emph{#1}}}%
   \index{#2}}

\newcommand{\emphi}[1]{\emphic{#1}{#1}}

\newlist{compactitem}{itemize}{4}
\setlist[compactitem,1]{nolistsep,label=$\star$,leftmargin=0.6cm}

\newlist{compactenumA}{enumerate}{5}%
\setlist[compactenumA]{topsep=0pt,itemsep=-1ex,partopsep=1ex,parsep=1ex,%
   label=(\Alph*)}%

\newlist{compactenuma}{enumerate}{5}%
\setlist[compactenuma]{topsep=0pt,itemsep=-1ex,partopsep=1ex,parsep=1ex,%
   label=(\alph*)}%

\newlist{compactenumi}{enumerate}{5}%
\setlist[compactenumi]{topsep=0pt,itemsep=-1ex,partopsep=1ex,parsep=1ex,%
   label=(\roman*)}%

\newcommand{\BadEvent}{\mathcal{B}}%
\newcommand{\Set}[2]{\left\{ #1 \;\middle\vert\; #2 \right\}}
\newcommand{\remove}[1]{}

\newcommand{\nSA}{\Mh{\mathcal{N}}}%

\newcommand{\SA}{\ensuremath{\Mh{U}}}%
\newcommand{\eA}{\ensuremath{\Mh{u}}}%

\newcommand{\ncols}{\Mh{k}}%

\newcommand{\SB}{\Mh{V}}
\newcommand{\eB}{\Mh{v}}%

\newcommand{\SC}{\Mh{B}}

\newcommand{\SX}{\Mh{X}}
\newcommand{\SQ}{\Mh{Q}}
\newcommand{\SD}{\Mh{W}}

\newcommand{\ds}{\displaystyle}%

\DefineNamedColor{named}{AlgorithmColor}{rgb}{0.1,0,0.15}

\newcommand{\AlgorithmI}[1]{{%
      \textcolor[named]{AlgorithmColor}{\texttt{\bf{#1}}}%
   }}
\newcommand{\Algorithm}[1]{{%
      \AlgorithmI{#1}%
      \index{algorithm!#1@{\AlgorithmI{#1}}}%
   }}

\newcommand{\CoarseEstimator}{\Algorithm{CoarseEstimator}\xspace}%
\newcommand{\CheckEstimate}{\Algorithm{Check{}Estimate}\xspace}%

\newcommand{\NiceFrac}[2]{%
   {\raisebox{0.2em}{\ensuremath{#1}}\!\!\!}%
   {\rotatebox[origin=c]{-20}{\Big/}}%
   {\! \raisebox{-0.2em}{\ensuremath{#2}} }}%

\newcommand{\Graph}{\Mh{G}}%
\newcommand{\Edges}{\Mh{E}}%
\newcommand{\AllEdges}{\Mh{Z}}%

\newcommand{\DS}{\Mh{\mathcal{D}}}%
\newcommand{\acc}{\Mh{\varphi}}%
\newcommand{\wtp}{\Mh{\mathsf{w}}}%

\newcommand{\EdgesX}[1]{\Mh{E}\pth{#1}}%
\newcommand{\EdgesY}[2]{\Mh{E}\pth{#1, #2}}%

\newcommand{\nEdges}{\Mh{\ensuremath{m}}}%
\newcommand{\mG}{\Mh{\nEdges(\Graph)}}%
\newcommand{\wX}[1]{\Mh{\overline{\mathsf{w}}}\pth{#1}}%
\newcommand{\mX}[1]{\nEdges\pth{#1}}%
\newcommand{\mY}[2]{\nEdges\pth{#1, #2}}%

\newcommand{\mAX}[1]{\nEdges_{\mathrm{active}}\pth{#1}}%

\newcommand{\ceil}[1]{\left\lceil {#1} \right\rceil}

\newcommand{\BadProb}{\Mh{\gamma}}%

\newcommand{\Sample}{\Mh{R}}%

\newcommand{\Entity}{\mathcal{H}}

\newcommand{\cA}{\Mh{{c}_1}}%
\newcommand{\cC}{\Mh{{c}_3}}%
\newcommand{\cE}{\Mh{{c}_5}}%

\newcommand{\Lsmall}{\Mh{L_{\mathrm{small}}}}%
\newcommand{\Llen}{\Mh{L_{\mathrm{len}}}}%
\newcommand{\Lbase}{\Mh{L_{\mathrm{base}}}}%

\newcommand{\RT}{\mathcal{R}}%
\newcommand{\degX}[1]{\Mh{\mathrm{deg}\pth{#1}}}%
\newcommand{\ctau}{\Mh{\varsigma}}

\newcommand{\pcdot}{\,}

\newcommand{\CC}{\mathcal{C}}%
\newcommand{\BFS}{breadth-first search\xspace}%
\newcommand{\VRT}{V}%
\newcommand{\VX}[1]{\VRT\pth{#1}}%

\makeatletter
\let\c@table\c@figure
\makeatother

\begin{document}

\title{Edge Estimation with Independent Set Oracles%
   \thanks{A preliminary version of this paper appeared in the proceedings of \si{ITCS}
      2018 \cite{bhrrs-eeiso-18}.}%
}

\date{}%

\author{%
   Paul Beame%
   \thanks{Paul G. Allen School Computer Science \& Engineering,
      University of Washington, Seattle.
      \url{beame@cs.washington.edu}}%
   \and%
   Sariel Har-Peled%
   \thanks{Department of Computer Science, University of Illinois,
      Urbana-Champaign. %
   }
   \and %
   Sivaramakrishnan Natarajan Ramamoorthy%
   \thanks{Paul G. Allen School Computer Science \& Engineering,
   	University of Washington, Seattle.%
      \url{sivanr@cs.washington.edu}%
   }%
   \and %
   Cyrus Rashtchian%
   \thanks{Department of Computer Science \& Engineering, UC San Diego. %
      \url{crashtchian@eng.ucsd.edu} }%
   \and %
   Makrand Sinha%
   \thanks{Centrum Wiskunde \& Informatika, Amsterdam, The Netherlands.  
      \url{makrand@cs.washington.edu}}%
}

\date{\today}%

\maketitle

\begin{abstract}  
    We study the task of estimating the number of edges in a graph,
    where the access to the graph is provided via an independent set
    oracle. Independent set queries draw motivation from group testing
    and have applications to the complexity of decision versus
    counting problems. We give two algorithms to estimate the number
    of edges in an $n$-vertex graph, using (i) $\polylog(n)$ bipartite
    independent set queries, or (ii) ${n}^{2/3} \pcdot\polylog(n)$
    independent set queries.
\end{abstract}

\section{Introduction}

We investigate the problem of estimating the number of edges in a
simple, unweighted, undirected graph $\Graph=(\IRX{n},\Edges)$, where
$\IRX{n} \deq \brc{1,2,\dots,n }$ and $\nEdges=|\Edges|$. Here, the
only access to the graph is provided via an oracle that answers
independent set queries. For a parameter $\e > 0$, we wish to output
an estimate $\estm$ satisfying
\begin{math}
(1-\e)\nEdges \leq \estm \leq (1+ \eps)\nEdges
\end{math}
with high probability. We consider randomized, adaptive algorithms
with access to one of the two following oracles:
\begin{compactitem}
	\smallskip%
	\item \emphi{\BIS} (Bipartite independent set) oracle: Given
	disjoint subsets $\SA, \SB \sq \IRX{n}$, a \BIS query answers
	whether there is no edge between $\SA$ and $\SB$ in
	$\Graph$. Formally, the oracle returns whether
	$\mY{\SA}{\SB} = 0$, where $\mY{\SA}{\SB}$ denotes the number of
	edges with one endpoint in $\SA$ and the other in~$\SB$.
	
	\smallskip%
	\item \emphi{\IS} (Independent set) oracle: Given a subset
	$\SA \sq \IRX{n}$, an \IS query answers whether $\SA$ satisfies
	$\mX{\SA}=0$, where $\mX{\SA}$ denotes the number of edges with
	both endpoints in $\SA$.
\end{compactitem}
\smallskip%
Previous work on graph parameter estimation has primarily focused on
{\em local} queries, such as
\begin{enumerate*}[label=(\roman*)]
	\item \emph{degree} queries (which output the degree of a vertex
	$v$),
	\item \emph{edge existence} queries (which answer whether a pair
	$\{u,v\}$ forms an edge), or
	\item \emph{neighbor} queries (which provide the $i$\th neighbor
	of a vertex $v$).
\end{enumerate*}
Feige~\cite{f-sirvu-06} and Goldreich and Ron~\cite{gr-aapg-08}
prove that there are cases where polynomial number of such  local queries are
required.%

These queries can only obtain \emph{local} information about the
graph.  This motivates an investigation of other types of natural
queries that may enable efficient parameter estimation. The
independent set queries described above generalize an {edge existence}
query, and their non-locality opens the door for sub-polynomial query
algorithms for various graph parameter estimation tasks.

\subsection{Motivation and related work}
\seclab{related}

The most relevant motivation for \BIS and \IS queries comes from the
area of sub-linear time algorithms for graph parameter
estimation. \BIS and \IS queries also have interesting connections to
the classical area of group testing, to emptiness versus counting
questions in computational geometry, and to the complexity of decision
versus counting problems.

\paragraph{Graph parameter estimation.}

Feige \cite{f-sirvu-06} showed how to use
$O\left({\sqrt{n}/\e}\right)$ degree queries to output $\estm$ that
satisfies $\nEdges \leq \estm \leq (2+\e)\nEdges$, where
$\nEdges = \cardin{\Edges}$. Moreover, he showed that any algorithm
achieving better than a 2-approximation must use a nearly linear
number of degree queries. Goldreich and Ron~\cite{gr-aapg-08} showed
that by using both degree and neighbor queries, the approximation
improves to
\begin{math}
(1-\e)\nEdges \leq \estm \leq (1+\e)\nEdges
\end{math}
by using
$O\bigl( (n/ \sqrt{\nEdges})\pcdot\poly(\log \nEdges,1/\e) \bigr)$
queries. It is worth noting that Feige \cite{f-sirvu-06} and Goldreich
and Ron \cite{gr-aapg-08} have identified certain hard instances
showing that these upper bounds cannot be improved, up to $\polylog$
factors.  Aliakbarpour \etal \cite{abgpry-stacs-18} do better than
the lower bound of Goldreich and Ron \cite{gr-aapg-08} by allowing one to
sample edges randomly.

Related work approximates the number of stars~\cite{grs-csoss-11}, the
minimum vertex cover~\cite{orrr-nosta-12}, the number of
triangles~\cite{s-ssaat-15,elrs-actst-17}, and the number of
$k$-cliques~\cite{ers-ankcs-17}. A special case of \BIS query (where
one of the bipartition sets is a singleton) has been used for testing
$k$-colorability of graphs \cite{bkkr-csqtp-13}, and high degree
vertex discovery \cite{wly-imcvh-13}.

\paragraph{Group testing.}
A classic estimation problem involves efficiently approximating the
number of defective items or infected individuals in a certain
collection or population~\cite{cs-ugtep-90, d-ddmlp-43,
	s-gteir-85}. To query a population, a small group is formed, and
all the individuals in the group are tested in one shot. For example,
in genome-wide association studies, combined pools of DNA may be
tested as a group for certain variants~\cite{kzcts-cfhpa-05}. In group
testing, the result of a test often indicates only whether there is at
least one infected or defective unit, or if there is none. Such a
dichotomous outcome resembles the \IS/\BIS queries.  In the graph
setting, group testing suggests testing pairwise interactions between
many items or individuals, instead of singular events.

\paragraph{Computational geometry.}
Certain geometric applications exhibit the phenomenon that emptiness
queries have more efficient algorithms than counting queries.  For
example, in three dimensions, for a set $P$ of $n$ points, half-space
counting queries (i.e., what is the size of the set $|P \cap h|$, for
a query half-space $h$), can be answered in $O(n^{2/3})$ time, after
near-linear time preprocessing.  On the other hand, emptiness queries
(i.e., is the set $P \cap h$ empty?) can be answered in $O( \log n)$
time. Aronov and Har-Peled \cite{ah-adrp-08} used this to show how to
answer approximate counting queries (i.e., estimating $|P \cap h|$),
with polylogarithmic emptiness queries.

As another geometric example, consider the task of counting edges in
disk intersection graphs using GPUs \cite{f-dgss-03}. For these
graphs, \IS queries decide if a subset of the disks have any
intersection (this can be done using sweeping in $O(n \log n)$
time~\cite{cj-spigu-15}). Using a GPU, one could quickly draw the
disks and check if the sets share a common pixel. In cases like this
-- when \IS and \BIS oracles have fast implementations -- algorithms
exploiting independent set queries may be useful.

\paragraph{Decision versus counting complexity.}
A generalization of \IS and \BIS queries previously appeared in a line
of work investigating the relationship between decision and counting
problems \cite{s-cac-83, s-aasp-85,
	dl-fgrac-18}. Stockmeyer~\cite{s-cac-83, s-aasp-85} showed how to
estimate the number of satisfying assignments for a circuit with
queries to an $\mathsf{NP}$ oracle.  Ron and Tsur~\cite{rt-pehss-16}
observed that Stockmeyer implicitly provided an algorithm for
estimating set cardinality using \emph{subset} queries, where a subset
query specifies a subset $X \sq \calU$ and answers whether
$|X \cap S| = 0$ or not.  Subset queries are significantly more
general and flexible than \IS and \BIS queries because $S$ corresponds
to the set of edges in the graph and $X$ is any subset of pairs of
vertices. Namely, \IS and \BIS queries can be interpreted as
restricted subset queries. In particular, the algorithms mentioned can
not be implemented directly using \IS or \BIS queries.

Indeed, consider subset queries in the context of estimating the
number of edges in a graph. To this end, fix $|S| = \nEdges$ (i.e.,
the number of edges in the graph) and $|\calU| = \binom{n}{2}$ (the
number of possible edges).  Stockmeyer provided an algorithm using
only $O(\log \log \nEdges \pcdot \poly(1/\eps))$ subset queries to
estimate $\nEdges$ within a factor of $(1 + \eps)$ with a constant
success probability. Note that for a high probability bound, which is
what we focus on in this paper, the algorithm would naively require
$O(\log n \cdot \log \log \nEdges \pcdot \poly(1/\eps))$ queries to
achieve success probability at least $1-1/n$.  Falahatgar \etal
\cite{fjops-endgt-16} gave an improved algorithm that estimates
$\nEdges$ up to a factor of $(1+\e)$ with probability $1-\delta$ using
$2\log \log \nEdges + O\pth{(1/\eps^2) \log(1/\delta)}$ subset
queries. Nearly matching lower bounds are also known for subset
queries~\cite{s-cac-83, s-aasp-85, rt-pehss-16, fjops-endgt-16}. Ron
and Tsur \cite{rt-pehss-16} also study a restriction of subset
queries, called {\em interval queries}, where they assume that the
universe $\calU$ is ordered and the subsets must be intervals of
elements.  We view the independent set queries that we study as
another natural restriction of subset queries.

Analogous to Stockmeyer's results, a recent work of Dell and Lapinskas
\cite{dl-fgrac-18} provides a framework that relates edge estimation
using \BIS and edge existence queries to a question in fine-grained
complexity. They study the relationship between decision and counting
versions of problems such as 3SUM and Orthogonal Vectors.  They proved
that, for a bipartite graph, using $O(\e^{-2}\log^6 n)$ \BIS queries,
and $\eps^{-4} n \pcdot\polylog(n)$ edge existence queries, one can
output a number $\estm$, such that, with probability at least
$1 - 1/n^2$, we have
\begin{math}
(1-\e)\nEdges\leq \estm \leq (1+\e)\nEdges.
\end{math}

Dell and Lapinskas \cite{dl-fgrac-18} used edge estimation to obtain
approximate counting algorithms for problems in fine-grained
complexity. For instance, given an algorithm for 3SUM with runtime~$T$, they obtain an algorithm that estimates the number of YES
instances of 3SUM with runtime
$ {O( T \eps^{-2}\log^6 n)} + {\eps^{-4} n\pcdot \polylog(n)}$. The
relationship is simple. The decision version of 3SUM corresponds to
checking if there is at least one edge in a certain bipartite
graph. The counting version then corresponds to counting the edges in
this graph. We note that in their application, the large number
$O(n \pcdot \polylog(n))$ of edge existence queries does not affect
the dominating term in the overall time in their reduction; the larger
term in the time is a product of the time to decide 3SUM and the
number of \BIS queries.

\subsection{Our results}
We describe two new algorithms.  Let $\Graph = (\IRX{n},\Edges)$ be a
simple graph with $\nEdges = \cardin{\Edges}$ edges.

\paragraph{The Bipartite Independence Oracle.}
We present an algorithm that uses \BIS queries and computes an
estimate $\estm$ for the number of edges in $\Graph$, such that
\begin{math}
(1- \eps)\nEdges \leq \estm \leq (1+\eps)\nEdges.
\end{math}
The algorithm performs $O( \eps^{-4} \log^{14} n )$ \BIS queries, and
succeeds with high probability (see \thmref{bisq} for a precise
statement). Ignoring the cost of the queries, the running time
is near linear (we mostly ignore running times in this paper, since query complexity is our main resource).  Since $\polylog(n)$ \BIS
queries can simulate a degree query (see \secref{d:est}), one can
obtain a $(2+\eps)$-approximation of~$\nEdges$ by using Feige's
algorithm~\cite{f-sirvu-06}, which uses degree queries.  This gives an
algorithm that uses
$O\pth{\sqrt{n} \pcdot \polylog(n)/\mathrm{poly}(\e)}$ \BIS
queries. Our new algorithm provides significantly better guarantees,
in terms of both the approximation and number of \BIS queries.

The result is somewhat more general than stated above. One can use
the algorithm to estimate the number of edges in any induced
subgraph of the original graph. Similarly, one can estimate the
number of edges in the graph between any two disjoint subsets of
vertices $\SA,\SB \subseteq \IRX{n}$. That is, the algorithm can
estimate the size of
$\EdgesY{\SA}{\SB}= \Set{ \eA \eB \in \Edges}{\eA \in \SA, \eB \in
	\SB}$.%

Compared to the result of Dell and Lapinskas \cite{dl-fgrac-18}, our
algorithm uses exponentially fewer queries, since we do not spend
$n \pcdot \polylog(n)$ edge existence queries. Our improvement does
not seem to imply anything for their applications in fine-grained
complexity. We leave open the question of finding problems where a
more efficient \BIS algorithm would lead to new decision versus
counting complexity results.

\paragraph{The Ordinary Independence Oracle.}
We also present a second algorithm, using only \IS queries to compute
a $(1+\eps)$-approximation. It performs
\begin{math}
O( \eps^{-4} \log^5 n + \min(n^2 / \nEdges, \sqrt{\nEdges}) \cdot
\eps^{-2} \log^2 n )
\end{math}
\IS queries (see \thmref{is:approx}). In particular, the number of \IS
queries is bounded by
\begin{math}
O( \eps^{-4} \log^5 n + \eps^{-2} n^{2/3} \log^2 n ).
\end{math}
The first term in the minimum (i.e., $\approx n^2/ \nEdges$) comes
from a folklore algorithm for estimating set cardinality using
membership queries (see \secref{subset:size}).  The second term in the
minimum (i.e., $\approx \sqrt{\nEdges}$) is the number of queries used
by our new algorithm.

We observe that \BIS queries are surprisingly more effective for
estimating the number of edges than \IS queries. Shedding light on
this dichotomy is one of the main contributions of this work.

\begin{table}[t]
	\centering
	\begin{tabular}{@{}lcll@{}}
		\toprule
		Query Types
		& Approximation
		& \# Queries {\footnotesize (up to \si{const.} factors) }
		& Reference \\
		\midrule
		Edge existence
		&     $1+\e$
		&  $\displaystyle\Bigl. (n^2/\nEdges)\ \poly(\log n, 1/\e)$
		& Folklore (see \secref{isq})
		\\
		Degree
		&     $2+\e$
		&  $\Bigl. \displaystyle{\sqrt{n}} \log n/{\e}$
		& \cite{f-sirvu-06}   %
		\\
		Degree + neighbor
		&     $1+\e$
		&  $\Bigl.\sqrt{n} \  \poly(\log n, 1/\e)$
		&  \cite{gr-aapg-08} %
		\\
		Subset
		& $1+\e$
		& $\Bigl.\poly( \log n,  1/\e)$
		& \cite{s-aasp-85, fjops-endgt-16}
		\\
		\BIS %
		& $1+\e$
		& $\Bigl. n \  \poly(\log n, 1/\e)$
		& \cite{dl-fgrac-18} %
		\\
		\midrule
		\BIS
		&     $1+\e$
		&  $\Bigl.\poly(\log n, 1/\e)$
		& This Work
		\\
		\IS
		&     $1+\e$
		&  $\Bigl. \min \pth{ \sqrt{\nEdges},\  \displaystyle {n^2}/{\nEdges} }
		\pcdot\poly(\log n, 1/\e)$
		& This Work
		\\ \bottomrule
	\end{tabular} 
	\caption{Comparison of the best known algorithms using a variety
		of queries for estimating the number of edges $m$ in a graph
		with $n$ vertices. The bounds stated are for high probability
		results, with error probability at most $1/n$. Constant factors
		are suppressed for readability.}
	\tablab{queries}
\end{table}

\paragraph{Comparison with other queries.}

\tabref{queries} summarizes the results for estimating the number of
edges in a graph in the context of various query types. Given some of
the results in \tabref{queries} on edge estimation using other types
of queries, a natural question is how well \BIS and \IS queries can
simulate such queries. In \secref{d:est}, we show that
$O( \eps^{-2} \log n)$ \BIS queries are sufficient to simulate degree
queries. On the other hand, we do not know how to simulate a neighbor
query (to find a specific neighbor) with few \BIS queries, but a
random neighbor of a vertex can be found with $O(\log n)$ \BIS queries
(see \cite{bkkr-csqtp-13}). For \IS queries, it turns out that
estimating the degree of a vertex $v$ up to a constant factor requires
at least $\Omega\left({n}/{\degX{v}}\right)$ \IS queries (see
\secref{is-degree}).

\paragraph{Notation.}
Throughout, $\log$ and $\ln$ denotes the logarithm taken in base $2$
and~$e$, respectively. For integers, $u, k$, let
$\IRX{k} = \{1,\dots,k\}$ and $\IRY{u}{k} = \{u,\dots,k\}$. The
notation $x = \polylog(n)$ means $x=O(\log^c n)$ for some constant
$c > 0 $. A collection of disjoint sets $\SA_1,\dots,\SA_k$ such that
$\bigcup_i \SA_i = \SA$, is a \emphi{partition} of the set $\SA$, into
$k$ \emphi{parts} (a part $\SA_i$ might be an empty set). In
particular, a (uniformly) \emphi{random partition} of $\SA$ into $k$
parts is chosen by coloring each element of $\SA$ with a random number
in $\IRX{k}$ and identifying $\SA_i$ with the elements colored with
$i$.

Throughout, we use $\Graph=(\IRX{n}, \Edges)$ to the denote the input
graph. The number of edges in $\Graph$ is denoted by
$\nEdges = \cardin{\Edges}$.  For a set $\SA \subseteq \IRX{n}$, let
$\EdgesX{\SA} = \Set{ uv \in \Edges}{u,v \in \SA}$ be the set of edges
between vertices of $\SA$ in $\Graph$.  For two disjoint sets
$\SA,\SB \subseteq \IRX{n}$, let $\EdgesY{\SA}{\SB}$ denote the set of
edges between $\SA$ and $\SB$:
$\EdgesY{\SA}{\SB}= \Set{ \eA \eB \in \Edges}{\eA \in \SA, \eB \in
	\SB}$. Let $\mX{\SA}$ and $\mY{\SA}{\SB}$ denote the number of
edges in $\EdgesX{\SA}$ and $\EdgesY{\SA}{\SB}$, respectively. We also
abuse notation and let $\nEdges(H)$ be the number of edges in a
subgraph $H$ (e.g., $\mG = \nEdges$).

\paragraph{High probability conventions.}
Through the paper, the randomized algorithms presented would succeed
with high probability; that is, with probability
$\geq 1 -1/n^{\Omega(1)}$.  Formally, this means the probability of
success is $\geq 1 - 1/n^c$, for some arbitrary constant $c > 0$. For
all these algorithms, the value of $c$ can be increased to any
arbitrary value (i.e., improving the probability of success of the
algorithm) by increasing the asymptotic running time of the algorithm
by a constant factor that depends only on $c$. For the sake of
simplicity of exposition, we do not explicitly keep track of these
constants (which are relatively well-behaved).

\subsection{Overview of the algorithms}
\seclab{overview}

\subsubsection{The \BIS algorithm}

Our discussion of the \BIS algorithm follows \figref{algorithm}, which
depicts the main components of one level of our recursive
algorithm. Our algorithms rely on several building blocks, as
described next.

\begin{figure}[p]
    \centering{\includegraphics[width=.99\textwidth]{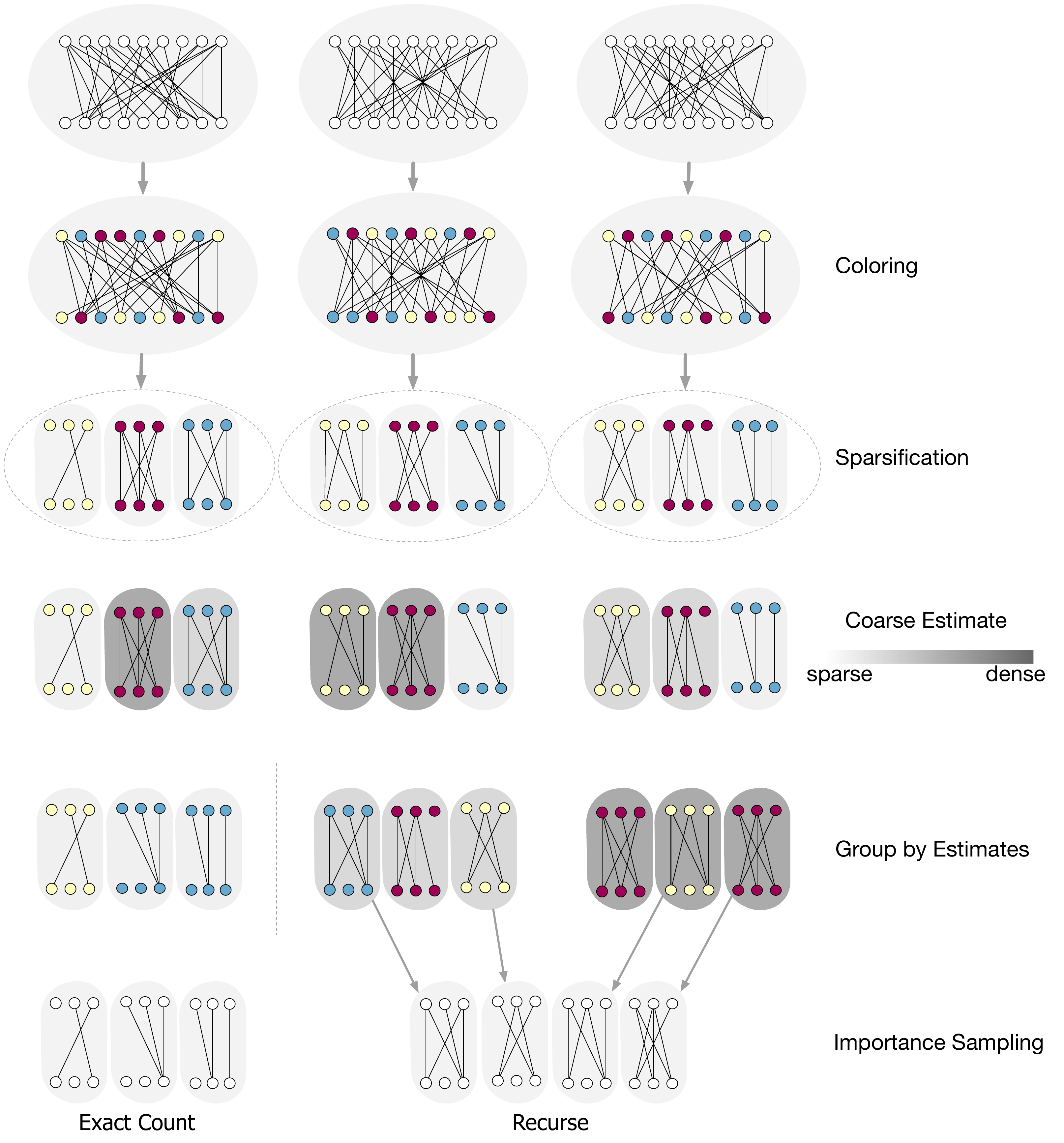}}
    \caption{A depiction of one level of the \BIS algorithm. In the
       first step, we color the vertices and sparsify the graph by
       only looking at the edges between vertices of the same
       color. In the second step, we coarsely estimate the number of
       edges in each colored subgraph. Next, we group these subgraphs
       based on their coarse estimates, and we subsample from the
       groups with a relatively large number of edges. In the final
       step, we exactly count the edges in the sparse subgraphs, and
       we recurse on the dense subgraphs.}
	\figlab{algorithm}
\end{figure}

\paragraph{Exactly count edges.}

One can exactly count the edges between two subsets of vertices, with
a number of queries that scales nearly linearly in the number of such
edges. Specifically, a simple deterministic divide and conquer
algorithm to compute $\mY{\SA}{\SB}$ using $O(\mY{\SA}{\SB} \log n)$
\BIS queries is described below in \lemref{bisq:exact}.

\paragraph{Sparsify.}
The idea is now to sparsify the graph in such a way that the number of
remaining edges is a good estimate for the original number of edges
(after scaling). Consider sparsifying the graph by coloring the
vertices of graph, and only looking at the edges going between certain
pairs of color classes (in our algorithm, these pairs are a matching
of the color classes). We prove that it suffices to only count the
edges between these color classes, and we can ignore the edges with
both endpoints inside a single color class.

For any $k$ satisfying $1 \leq k \leq \lfloor n/2\rfloor$, let
$\SA_1,\dots,\SA_k, \SB_1,\dots,\SB_k$ be a uniformly random partition
of~$\IRX{n}$. Then, we have
\begin{equation}
\Prob{ \biggl. \Bigl| \smash{2k \sum_{i=1}^{k} \mY{\SA_i}{\SB_i} -
		\nEdges }\Bigr|  \geq
	c k  \sqrt{\nEdges} \log n } \leq \frac{1}{n^4}, 
\smallskip%
\end{equation}
\noindent%
where $c$ is some constant.  For the proof of this inequality see
\secref{sparsification}.  Specifically, if we set $\Graph_i$ to be the
induced bipartite subgraph on $\SA_i$ and $\SB_i$, then
$2k \sum_{i} \mX{G_i}$ is a good estimate for $\mG$.

\paragraph{Now the graph is bipartite.}
The above sparsification method implies that we can assume without
loss of generality that the graph is bipartite. Indeed, invoking the
lemma with $k=1$, we see that estimating the number of edges between
the two color classes is equivalent to estimating the total number of
edges, up to a factor of two. For the rest of the discussion, we will
consider colorings that respect the bipartition.

\paragraph{Coarse estimator.}
We give an algorithm that coarsely estimates the number of edges in a
(bipartite) subgraph, up a $O(\log^2 n)$ factor, using only
$O(\log^3 n)$ \BIS queries.

\paragraph{The subproblems.}
After coloring the graph, we have reduced the problem to estimating
the total number of edges in a collection of (disjoint) bipartite
subgraphs. However, certain subgraphs may still have a large number of
edges, and it would be too expensive to directly use the exact
counting algorithm on them.

\paragraph{Reducing the number of subgraphs in a collection, via
	importance sampling.}
Using the coarse estimates we can form $O(\log n)$ groups of bipartite
subgraphs, where each group contains subgraphs with a comparable
number of edges.  For the groups with only a polylogarithmic number of
edges, we can exactly count edges using $\polylog(n)$ \BIS queries via
the exact count algorithm mentioned above. For the remaining groups,
we subsample a polylogarithmic number of subgraphs from each group.
This new estimate is a good approximation to the original quantity,
with high probability. This corresponds to the technique of {\em
	importance sampling} that is used for variance reduction when
estimating a sum of random variables that have comparable magnitudes.

\paragraph{Sparsify and reduce.}
We use the sparsification algorithm on each graph in our
collection. This increases the number of subgraphs while reducing (by
roughly a factor of $k$) the total number of edges in these
graphs. The number of edges in the new collection is a reliable
estimate for the number of edges in the old collection. We will choose
$k$ to be a constant so that every sparsification round reduces the
number of edges by a constant factor.

If the number of graphs in the collection becomes too large, then we
reduce it in one of two ways. For the subgraphs with relatively few
edges, we exactly count the number of edges using only $\polylog(n)$
queries. For the dense subgraphs, we can apply the above importance
sampling technique and retain only $\polylog(n)$ subgraphs. Every
basic operation in this scheme requires $\polylog(n)$ \BIS queries,
and the number of subgraphs is $\polylog(n)$. Therefore, a round can
be implemented using $\polylog(n)$ \BIS queries. Now, since every
round reduces the number of edges by a constant factor, the algorithm
terminates after $O( \log n)$ rounds, resulting in the desired
estimate for $\nEdges$ using only $\polylog(n)$ queries in
total. \figref{algorithm} depicts the main components of one round.

\medskip%

We have glossed over some details regarding the reweighting of
intermediate estimates, as both the sparsfication and importance
sampling steps involve subsampling and rescaling. To handle this, the
algorithm will maintain a weight value for each subgraph in the
collection (starting with unit weight). Then, these weights will be
updated throughout the execution, and they will be used during coarse
estimation. For the final estimate, the algorithm will output a
weighted sum of the estimates for the remaining subgraphs, in addition
to the weighted version of the exactly counted subgraphs. By using
these weights to properly rescale estimates and counts, the algorithm
will achieve a good estimate for $\nEdges$ with high probability.

\subsubsection{The \IS algorithm}
We move on to describe our second algorithm, based on $\IS$
queries. As with the \BIS algorithm, the main building block for the
\IS algorithm is an efficient way to exactly count edges using \IS
queries. The exact counting algorithm works by first breaking the
vertices of the graph into independent sets in a greedy fashion, and
then grouping these independent sets into larger independent sets
using (yet again) a greedy algorithm. The resulting partition of the
graph into independent sets has the property that every two sets have
an edge between them, and this partition can be computed using a
number of queries that is roughly $\nEdges$. This is beneficial,
because when working on the induced subgraph on two independent sets,
the \IS queries can be interpreted as \BIS queries. As such, edges
between parts of the partition can be counted using the exact counting
algorithm, modified to use \IS queries.  The end result is, that for a
given set $\SA \subseteq \IRX{n}$, one can compute $\mX{\SA}$, the
number of edges with both endpoints in $\SA$, using
$O( \mX{\SA} \log n)$ \IS queries. This algorithm is described in
\secref{quadtree}.

Now, we can sparsify the graph to reduce the overall number of \IS
queries. In contrast to the \BIS queries, we do not know how to design
a coarse estimator using only \IS queries (see \secref{limit}). This
prohibits us from designing a similar algorithm. Instead, we estimate
the number of edges in one shot, by coloring the graph with a large
number of colors and estimating the number of edges going between a
matching of the color classes.

somewhat counter intuitive.  An initial sparsification attempt might
be to count only the edges going between a single pair of colors. If
the total number of colors is $2k$, then we expect to see
${\nEdges}/{\binom{2k}{2}}$ edges between this pair. Therefore, we
could set $k$ to be large and invoke \lemref{isq:exact}. Scaling by a
factor of ${\binom{2k}{2}}$, we would hope to get an {\em unbiased}
estimator for $\nEdges$.

Unfortunately, a star graph demonstrates that this approach does not
work, due to the large variance of this estimator. If we randomly
color the vertices of the star graph with $2k$ colors, then out of the
$\binom{2k}{2}$ pairs of color classes, only $2k-1$ pairs have any
edge going between color classes. So, if we only chose one pair of
color classes, then with high probability one of the following two
cases occurs: either (i) there is no edge crossing the color pair, or
(ii) the number of edges crossing the pair is $\approx \nEdges/2k$. In
both cases our estimate after scaling by a factor of $\binom{2k}{2}$
will be far from the truth.

At the other extreme, the vast majority of edges will be present if we
look at the edges crossing {\em all pairs} of color classes. Indeed,
the only edges we miss have both endpoints in a color class, and this
accounts for only a $1/2k$ fraction of the total number of
edges. Thus, this does not achieve any substantial sparsification.

By using a matching of the color classes, we simultaneously get a
reliable estimate of the number of edges and a sufficiently sparsified
graph (see \lemref{sparsification_general}).  Let
$\SA_1,\ldots,\SA_k,\SB_1,\ldots,\SB_k$ be a random partition of the
vertices into $2k$ color classes.  This implies that with high
probability, the estimator $2k \sum_{i=1}^{k} \mY{\SA_i}{\SB_i}$ is in
the range $\nEdges \pm O(k\sqrt{\nEdges}\log n)$. Hence, as long as we
choose $k$ to be less than $ \e \sqrt{\nEdges}/\polylog(n)$, we
approximate $\nEdges$ up to a factor of $(1 + O(\e))$.  We use
geometric search to find such a $k$ efficiently.

To get a bound on the number of \IS queries, we claim that we can
compute $\sum_{i=1}^{k} \mY{\SA_i}{\SB_i}$ using \lemref{isq:exact},
with a total of
$\left(k + \frac{\nEdges}{k}\right) \pcdot \polylog (n)$ \IS
queries. The first term arises since we have to make at least one
query for each of the $k$ color pairs (even if there are no edges
between them). For the second term, we pay for both (i) the edges
between the color classes and (ii) the total number of edges with both
endpoints within a color class (since the number of \IS queries in
\lemref{isq:exact} scales with $\mX{\SA \cup \SB})$. By the
sparsification lemma, we know that (i) is bounded by
$O({\nEdges}/{k})$ with high probability and we can prove an analogous
statement for (ii). Hence, plugging in a value of
$k \approx { \e \sqrt{\nEdges}}/{\polylog(n)}$, the total number of
\IS queries is bounded by $\sqrt{\nEdges}\pcdot \polylog(n) / \e$.

\subsection{Subsequent work after initial publication}

After the initial publication of our results~\cite{bhrrs-eeiso-18},
there has been some follow-up work \cite{bbgm-teuti-19, bbgm-heups-19,
	clw-noeei-20, dlm-acssw-20}.

Answering one of the open questions of \cite{bhrrs-eeiso-18}, Chen, Levy, and Waingarten~\cite{clw-noeei-20} provide nearly-matching upper and lower
bounds on the number of \IS queries for edge estimation. More
precisely, they show that
\begin{math}
O\bigl( \min ( n/\sqrt{\nEdges}, \ \sqrt{\nEdges} \, ) \cdot \poly(\log(n),
1/\e) \bigr)
\end{math}
\IS queries are sufficient (the term $n^2/\nEdges$ is the new
result). They also prove that
$\Omega\bigl( \min ( n/\sqrt{\nEdges},\ \sqrt{\nEdges} \, )/ \polylog(n)
\bigr)$ \IS queries are necessary for a certain family of graphs.

Dell, Lapinksas, and Meeks~\cite{dlm-acssw-20} provide new connections between
decision and approximate counting results for problems such as
$k$-SUM, $k$-Orthogonal-Vectors, and $k$-Clique, by relating the
complexity to edge estimation using certain queries. In particular,
their work extends the previous work of Dell and
Lapinskas~\cite{dl-fgrac-18} to the case of $k$-hypergraphs, and they
consider a generalization of \BIS queries to $k$-partite set
queries. As one of their technical results, they improve the
dependence on $\e$ in \thmref{bisq} from $\e^{-4}$ down to $\e^{-2}$.

Bhattacharya, Bishnu, Ghosh, and Mishra~\cite{bbgm-teuti-19, bbgm-heups-19} also consider
the generalization of \BIS queries to tripartite set queries, where
they use such queries to estimate the number of triangles in a graph.

\subsection{Outline}
The rest of the paper is organized as follows.
We start at \secref{prelim} by reviewing some necessary tools --
concentration inequalities, importance sampling, and set size
estimation via membership queries.
In \secref{sparsification}, we prove our sparsification result
(\lemref{sparsification_general}).

In \secref{bisq} we describe the algorithm for edge estimation for the
\BIS case. \secref{bis:exact} describes the exact counting algorithm.
In \secref{coarse}, we present the algorithm that uses \BIS queries to
coarsely estimate the number of edges between two subsets of vertices
(\lemref{coarse_estimate}). We combine these building blocks to
construct our edge estimation algorithm using \BIS queries in
\secref{bisalgo}.

The case of \IS queries is tackled in \secref{is:e:e}.  In
\secref{quadtree}, we formally present the algorithms to exactly count
edges between two subsets of vertices (\lemref{isq:exact}).  In
\secref{isq}, we present our algorithm using \IS queries.  In
\secref{limit} we provide some discussion of why the \IS case seems to
be harder than the \BIS case.
We conclude in \secref{conclusion} and discuss open questions.

\section{Preliminaries}
\seclab{prelim}

Here we present some standard tools that we need later on.

\subsection{Concentration bounds}

For proofs of the following concentration bounds, see the book by
Dubhashi and Panconesi~\cite{dp-cmara-09}.

\begin{lemma}[Hoeffding's inequality]
	\lemlab{Hoeffding:inequality}%
	Let $X_1, \dots, X_r$ be independent random variables satisfying
	$X_i \in [a_i, b_i]$ for $i \in [r]$.  Then, for
	$ X = X_1 + \cdots + X_r$ and any $s > 0$, we have
	\begin{math}
	\Prob{ \cardin{ X - \Ex{X} } \geq s} \leq 2\exp \pth{ -
		{2\,s^2}/{\sum_{i=1}^r (b_i - a_i)^2} }.
	\end{math}
\end{lemma}

\begin{lemma}[{{Chernoff-Hoeffding inequality \protect\cite[Theorem
			1.1]{dp-cmara-09}}}]%
	\lemlab{chernoff}%
	Let $X_1, \dots, X_r$ be $r$ independent random variables with
	$0 \leq X_i \leq 1$, and let $X = \sum_{i=1}^r X_i$. For
	$\mu = \Ex{X}$, let $\ell$ and $u$ be real numbers such that
	$\ell \leq \mu \leq u$. Then, we have that \smallskip%
	\begin{compactenumA}[itemsep=-0.3ex]
		\item \itemlab{Chernoff:A} %
		For any $\Delta >0$, we have
		\begin{math}
		\Prob{X \leq \ell - \Delta} \leq%
		\exp\pth{-2\Delta^2/r}
		\end{math}
		and $\Prob{X \geq u + \Delta} \leq \exp\pth{-2\Delta^2/r}$.
		
		\item \itemlab{Chernoff:B} %
		For any $0 \leq \delta < 1$, we have
		\begin{math}
		\Prob{X \leq (1-\delta) \mu}%
		\leq%
		\exp \pth{-{\mu \delta^2}/{2}}.
		\end{math}
		\item \itemlab{Chernoff:C} %
		For any $0 \leq \delta \leq 1$, we have
		\begin{math}
		\Prob{X \geq (1+\delta) \mu}%
		\leq%
		\exp \pth{-{\mu \delta^2}/{3}}.
		\end{math}
	\end{compactenumA}
\end{lemma}

We need a version of Azuma's inequality that takes into account a rare
bad event -- the following is a restatement of Theorem 8.3 from
\cite{cl-cimis-06} in a simplified form (that is sufficient for our
purposes).

\begin{lemma}[{\cite{cl-cimis-06}}]
	\lemlab{azuma:cond}%
	Let $f$ be any function of $r$ independent random variables
	$Y_1,\ldots,Y_r$, and let
	$X_i = \ExCond{f(Y_1, \ldots, Y_r)}{Y_1, \ldots, Y_i}$, for
	$i\in \IRX{r}$, and $X_0 = \Ex{f(Y_1, \ldots, Y_r)}$.  Say that a
	sequence $Y_1, \ldots, Y_r$ is \emph{bad} if there exists an index
	$i$ such that $\cardin{X_i - X_{i-1}} > c_i$, where
	$c_1, \ldots, c_r$ are some nonnegative numbers. Let $\BadEvent$
	be the event that a bad sequence happened, and let
	$S = \sum_{i=1}^r c_i^2$. We have that
	\begin{math}
	\Prob{ \bigl. \cardin{X_r - X_0} \geq \lambda }%
	\leq%
	2\exp \pth{-\NiceFrac{\lambda^2}{2S}} + \Prob{\BadEvent}.
	\end{math}
\end{lemma}

\subsection{Importance sampling}

Importance sampling is a technique for estimating a sum of
terms. Assume that for each term in the summation, we can cheaply and
quickly get an initial, coarse estimate of its value. Furthermore,
assume that better estimates are possible but expensive. Importance
sampling shows how to sample terms in the summation, then acquire a
better estimate {\em only for the sampled terms}, to get a good
estimate for the full summation. In particular, the number of samples
is bounded independently of the original number of terms, depending
instead on the coarseness of the initial estimates, the probability of
success, and the quality of the final output estimate.

\begin{lemma}[Importance Sampling]
	\lemlab{importance}%
	Let $U = \brc{u_1, \ldots, u_r}$ be a set of numbers, all
	contained in the interval $[\alpha/b, \alpha b]$, for $\alpha > 0$
	and $b \geq 1$. Let $\gamma, \eps > 0$ be parameters. Consider the
	sum $\Gamma = \sum_{i=1}^r u_i$.  For an arbitrary
	\begin{math}
	t \geq \frac{b^4}{2\eps^2} \bigl( 1 + \ln \frac{1}{\gamma}
	\bigr),
	\end{math}
	and $i=1,\ldots, t$, let $X_i$ be a random sample chosen uniformly
	(and independently) from the set $U$ (i.e., let $j_i$ be uniformly
	and randomly picked from $\IRX{r}$, and let $X_i =
	u_{j_i}$). Then, the estimate $Y = (r/t)\sum_{i=1}^t X_i$ for the
	value of $\Gamma$ satisfies
	\begin{math}
	\Prob{\cardin{Y - \Gamma} \geq \eps \Gamma } \leq \gamma.
	\end{math}
\end{lemma}

\begin{proof}
	Observe that $r(\alpha/b) \leq \Gamma \leq r \alpha b$, and
	\begin{equation*}
		\mu%
		=%
		\Ex{Y}%
		=%
		\Ex{\Bigl. \smash{(r/t)\sum_{i=1}^t X_i}}%
		=%
		(r/t)\sum_{i=1}^t
		\Ex{X_i}%
		=%
		\frac{r}{t} \cdot t \cdot \frac{\Gamma}{r}= \Gamma.
	\end{equation*}    
	Furthermore, we have $Z = ({t}/{r})Y = \sum_{i=1}^t X_i$,
	$\Ex{Z} = (t/r)\Gamma$, and for the length, $\Delta_i$, of the
	interval containing $X_i$, we have
	\begin{math}
	\Delta_i^2 = (\alpha b - \alpha/b )^2 \leq \alpha^2 b^2.
	\end{math}
	
	Using $r \alpha/b \leq \Gamma$, by \lemref{Hoeffding:inequality},
	we have
	\begin{align*}
		\Prob{\cardin{Y - \Gamma} \geq \eps \Gamma \Bigr.}%
		&=%
		\Prob{\cardin{\frac{t}{r}Y - \frac{t}{r}\Gamma}
			\geq \frac{t \eps \Gamma}{r} }%
		\leq%
		\Prob{\cardin{\sum_{i=1}^t X_i - \frac{t}{r}\Gamma}
			\geq \frac{t \eps r\alpha/b}{r} }%
		\\&
		=%
		\Prob{\cardin{\sum_{i=1}^t X_i - \Ex{Z}}
			\geq \frac{t \eps \alpha}{b} }%
		\leq%
		2\exp
		\pth{ - \frac{2\,({t \eps \alpha}/{b})^2}{\sum_{i=1}^t
				\Delta_i^2 }}%
		\leq%
		2\exp \pth{ - \frac{2\,({t \eps \alpha}/{b})^2}{t \alpha^2 b^2}}%
		\\&
		=%
		2\exp\pth{ - \frac{2 t \eps^2  }{b^4  }}%
		\leq%
		\gamma.
	\end{align*}
\end{proof}

The above lemma enables us to reduce a summation with many numbers
into a much shorter summation (while introducing some error,
naturally). The list/summation reduction algorithm we need is
described next.

\begin{lemma}[Summation reduction]
	\lemlab{importance:alg}%
	Let $(\Entity_1, w_1,e_1), \ldots, (\Entity_r,w_r,e_r)$ be given,
	where $\Entity_i$'s are some structures, and $w_i$ and $e_i$ are
	numbers, for $i=1, \ldots, r$. Every structure $\Entity_i$ has an
	associated unknown cost $\wX{\Entity_i} \geq 0$. The quantity of
	interest, that we would like to compute/approximate is
	\begin{equation*}
		\Gamma = \sum_i w_i \cdot \wX{\Entity_i}.
	\end{equation*}
	To this end, we have parameters $\epsA > 0$, $\BadProb$, $b$, and
	$M$, such that: \smallskip%
	\begin{compactenumi}[leftmargin=3em]
		\item $\forall i \quad w_i,e_i \geq 1$, \smallskip%
		\item $\forall i \quad e_i/b \leq \wX{\Entity_i} \leq e_i
		b$, %
		and \smallskip%
		\item $\Gamma \leq M$
	\end{compactenumi}
	\smallskip%
	Then, one can compute a new (hopefully shorter) sequence of
	triples
	$(\Entity_1', w_1', e_1'), \ldots, (\Entity_t',w_t', e_t')$ (the
	new sequence is a subsequence of the original sequence with
	reweighting). The new sequence complies with the above conditions,
	and furthermore, the estimate
	\begin{equation*}
		Y = \sum_{i=1}^t w_i' \wX{\Entity_i'}        
	\end{equation*}
	is a multiplicative $(1\pm \epsA)$-approximation to $\Gamma$, with
	probability $\geq 1- \BadProb$.  The running time of the algorithm
	is $O( r)$, and size of the output sequence is
	\begin{math}
	t = O\pth{ b^4 \epsA^{-2} (\log \log M + \log \BadProb^{-1})
		\log M }.
	\end{math}
\end{lemma}
\begin{proof}
	We break the interval $[1, M]$ into $\log M$ intervals in the
	natural way, where the $j$\th interval is
	$J_j = \bigl[2^{j-1}, 2^j\bigr)$, for
	$j=1, \ldots, h = \ceil{\log M}$, except if $M$ is a power of 2,
	in which case the last interval is closed and also includes
	$2^h=M$.  Input triples are sorted into $h$ groups
	$U_1, \ldots, U_h$, where an input triple $(\Entity, w, e)$ is in
	$U_j$, if $e w \in J_j$. This mapping can be done in $O(r)$ time.
	
	Let
	$\alpha = O\bigl( b^4 \epsA^{-2}\bigl[ 1+ \ln (h/\BadProb) \bigr]
	\bigr) $.  For $j=1, \ldots, h$, if $\cardin{U_j} \leq \alpha$,
	then set $\Sample_j = U_j$, otherwise compute a sample $\Sample_j$
	from $U_j$ of size $\alpha$.  We associate weight
	$W_j = \cardin{U_j} / \cardin{\Sample_j}$ with $\Sample_j$. If a
	triple $(\Entity, w, e) \in U_j$, then we have that
	$w\cdot \wX{\Entity} \in \bigl[2^{j-1}/b, 2^j b\bigr]$.

	For all $j\in \IRX{h}$, let
	\begin{math}
	\Gamma_j = \sum_{(\Entity,w,e) \in U_j} w\cdot
	\wX{\Entity}\Bigr.
	\end{math}
	be the total weight of structures in the $j$\th group.  By
	\lemref{importance}, we have, with probability
	$\geq 1 - \BadProb/h$, that
	\begin{equation*}
		Y_j = \Bigl( W_j \sum\nolimits_{(\Entity,w,e) \in \Sample_j}
		\wX{\Entity}\cdot w \Bigr)%
		\in%
		\bigl[(1-\epsA) \Gamma_j, (1+\epsA)\Gamma_j\bigr].
	\end{equation*}
	Summing these inequalities over all $j\in \IRX{h}$, implies that
	$Y$ is the desired approximation with probability
	$\geq 1 -\BadProb$.
	
	Specifically, the output sequence is constructed as follows. For
	all $j\in \IRX{h}$, and for every triple
	$(\Entity, w, e) \in \Sample_j$, we add $(\Entity, w\cdot W_j, e)$
	to the output sequence. Clearly, the output sequence has
	\begin{math}
	t = h \alpha%
	=%
	O\bigl(%
	b^4 \epsA^{-2} (\log \log M + \log \BadProb^{-1}) \log M
	\bigr)
	\end{math}
	elements.
\end{proof}

\begin{remark:unnumbered}
	(A) The algorithm of \lemref{importance:alg} does not use the
	entities $\Entity_i$ directly at all. In particular, the
	$\Entity_i'$s are just copies of some original structures. The
	only thing that the above lemma uses is the estimates
	$e_1, \ldots, e_r$ and the weights $w_1, \ldots, w_r$.
	
	(B) The sampling size used in \lemref{importance:alg} can probably
	be improved by a polylog factors by sampling directly from all
	$\log M$ classes simultaneously.
\end{remark:unnumbered}

\begin{remark}
	\remlab{L:2}%
	We are going to use \lemref{importance:alg}, with
	$\epsA = O(\eps/ \log n)$, $\BadProb = 1/n^{O(1)}$,
	$b = O(\log n)$, and $M = n^2$. As such, the size of the output
	list is
	\begin{equation*}
		\Llen%
		=%
		O\pth{ \log^4 n \cdot \eps^{-2} \log^2 n \cdot (\log
			\log n + \log n) \log n }%
		=%
		O( \eps^{-2} \log^8 n ).
	\end{equation*}
\end{remark}

\subsection{Estimating subset size via membership oracle %
	queries}
\seclab{subset:size}%

We present here a standard tool for estimating the size of a subset
via membership oracle queries. This is well known, but we provide the
details for the sake of completeness.

\begin{lemma}
	\lemlab{estimate}%
	Consider two (finite) sets $\SC \subseteq \SA$, where
	$n = \cardin{\SA}$.  Let $\eps\in (0,1)$ and
	$\BadProb \in (0,1/2)$ be parameters.  Let $g > 0$ be a
	user-provided guess for the size of $\cardin{\SC}$. Consider a
	random sample $\Sample$, taken with replacement from $\SA$, of
	size $r = \ceil{\cE \eps^{-2} (n/g) \log \BadProb^{-1}}$, where
	$\cE$ is sufficiently large. Next, consider the estimate
	$Y = (n/r) \cardin{ \Sample \cap \SC}$ to $\cardin{\SC}$. Then, we
	have the following: \smallskip%
	\begin{compactenumA}
		\item If $Y < g/2$, then $\cardin{\SC} < g$,
		\item If $Y \geq g/2$, then
		$(1-\eps)Y \leq \cardin{\SC} \leq (1+\eps)Y$.
	\end{compactenumA}
	\smallskip%
	Both statements above hold with probability $\geq 1-\BadProb$.
\end{lemma}
\begin{proof}
	
	(A) The bad scenario here is that $\cardin{\SC} \geq g$, but
	$Y < g/2$. Let $X_i=1$ $\iff$ the $i$\th sample element is in
	$\SC$. We have that $Y = (n/r) X$, where $X= \sum_{i=1}^r X_i$.
	By assumption, we have
	\begin{equation}
	\mu = %
	\Ex{X}%
	=%
	\frac{r \cardin{\SC}}{n}%
	\geq%
	\frac{rg}{n}%
	\geq %
	\cE \frac{  \log \BadProb^{-1}}{ \eps^2 }.
	\eqlab{mu}%
	\end{equation}
	As such, by Chernoff's inequality (\lemref{chernoff}
	\itemref{Chernoff:B}), we have that
	\begin{math}
	\Prob{ Y < g/2} %
	=%
	\Prob{ X < rg/(2n) } %
	=%
	\ProbChar \bigl[ X < (1-1/2) \Ex{X} \bigr]%
	\leq%
	\exp\pth{ -\mu/8 } %
	\leq%
	\BadProb^{\cE /(8\eps^2\ln 2)}
	\end{math}
	and this is $\leq \BadProb$ for $\cE$ a sufficiently large
	constant.
	
	\medskip\noindent%
	(B) We have two cases to consider.  First suppose that $|B|< g/4$.
	In this case, if $X=\sum_{i=1}^r X_i$ is the random variable as
	described part (A), then each $X_i$ is an indicator variable with
	probability $p=|B|/n < g/(4n)$ and
	$\Prob{Y\ge g/2}=\Prob{X\ge rg/(2n)}\le \Prob{X'\ge rg/(2n)}$
	where $X'$ is the sum of $r$ independent Bernoulli trials with
	success probability $g/(4n)$.  Now
	$\Ex{X'}=\frac{rg}{4n}\ge \cE \frac{\log \BadProb^{-1}}{4\eps^2}$
	so
	\begin{equation*}
		\Prob{Y\ge g/2}%
		\leq%
		\Prob{X'\ge rg/(2n)}%
		=%
		\Prob{X'\ge (1+1) \Ex{X'}\bigr.}
		\leq%
		\exp\pth{-\Ex{X'}/3}\le \BadProb^{\cE/(12\eps^2\ln 2)}        
	\end{equation*}
	by Chernoff's inequality (\lemref{chernoff} \itemref{Chernoff:C})
	and again this is $\leq \BadProb$ for $\cE$ a sufficiently large
	constant.
	
	For the second case, suppose that $|B|\ge g/4$.  Then,
	$\Ex{X}\ge \Ex{X'}\ge \cE \frac{\log \BadProb^{-1}}{4\eps^2}$ and,
	since $Y$ is a fixed multiple of $X$, by Chernoff's inequality
	(\lemref{chernoff} \itemref{Chernoff:B}), we have
	\begin{equation*}
		\Prob{\bigl. Y < (1-\eps)\Ex{Y}} %
		=%
		\Prob{ \bigl. X < (1-\eps)\Ex{X} } %
		\leq%
		\exp\pth{ -\Ex{X} \eps^2/2 } %
		\leq%
		\BadProb^{\cE/(8\ln 2)}
	\end{equation*}
	which is $\le \BadProb/2$ for $\cE \geq 16\ln 2$.  Similarly, by
	Chernoff's inequality (\lemref{chernoff} \itemref{Chernoff:C}),
	\begin{equation*}
		\Prob{\bigl. Y > (1+\eps)\Ex{Y}} %
		=%
		\Prob{ \bigl. X > (1+\eps)\Ex{X} } %
		\leq%
		\exp\pth{ -\Ex{X} \eps^2/3 } %
		\leq%
		\BadProb^{\cE/(12\ln 2)}
	\end{equation*}
	which is $\le \BadProb/2$ for $\cE \geq 24 \ln 2$, as
	$\BadProb \leq 1/2$.  Adding these two failure probabilities
	together gives a bound of at most $\BadProb$ as required.
\end{proof}

\begin{lemma}
	\lemlab{est:set}%
	Consider two sets $\SC \subseteq \SA$, where $n = \cardin{\SA}$.
	Let $\epsA, \BadProb \in (0,1)$ be parameters, such that
	$\BadProb < 1/ \log n$. Assume that one is given an access to a
	membership oracle that, given an element $x \in \SA$, returns
	whether or not $x \in \SC$. Then, one can compute an estimate $s$,
	such that
	$(1-\epsA)\cardin{\SC}\leq s \leq (1+\epsA)\cardin{\SC}$, and
	computing this estimate requires
	$O( (n/\cardin{\SC}) \epsA^{-2} \log \BadProb^{-1})$ oracle
	queries. The returned estimate is correct with probability
	$\geq 1 - \BadProb$.
\end{lemma}
\begin{proof}
	Let $g_i = n/2^{i+2}$. For $i=1,\ldots, \log n$, use the algorithm
	of \lemref{estimate} with $\eps = 0.5$, with the probability of
	failure being $\BadProb / (8\log n)$, and let $Y_i$ be the
	returned estimate. The algorithm stops this loop as soon as
	$Y_i \geq 4 g_i$. Let $I$ be the value of $i$ when the loop
	stopped. The algorithm now calls \lemref{estimate} again with
	$g_I$ and $\eps=\epsA$, and returns the value of $Y$, as the
	desired estimate.

	Overall, for $T = 1 + \ceil{ \log n }$, the above makes $T$ calls
	to the subroutine of \lemref{estimate}, and the probability that
	any of them to fail is $T \BadProb / (8 \log n) <
	\BadProb$. Assume that all invocations of \lemref{estimate} were
	successful. In particular, \lemref{estimate} guarantees that if
	$Y > 4g_I \geq g_I/2$, then the estimate returned is
	$( 1\pm \eps)$-approximation to the desired quantity.
	
	Computing $Y_i$ requires
	\begin{math}
	r_i = O((n/g_i) \log (\log n/\BadProb )) = O(2^i \log
	\BadProb)
	\end{math}
	oracle membership queries.  As such, the number of membership
	queries performed by the algorithm overall is
	\begin{align*}
		\smash{\sum_i} r_i + O( (n/g_I) \eps^{-2} \log (\log n/\BadProb ) )%
		=%
		O( (n/\cardin{\SC}) \eps^{-2} \log \BadProb ).
	\end{align*}
\end{proof}

\subsubsection{Estimating subset size via emptiness oracle %
	queries}

Consider the variant where we are given a set $\SX \subseteq \SA$.
Given a query set $\SQ \subseteq \SA$, we have an \emph{emptiness
	oracle} that tells us whether $\SQ \cap \SX$ is empty. Using an
emptiness oracle, one can get a $(1\pm \eps)$-approximate the size of
$\SX$ using relatively few queries. The following result is implied by
the work of Aronov and Har-Peled \cite[Theorem 5.6]{ah-adrp-08} and
Falahatgar \etal \cite{fjops-endgt-16} -- the latter result has better
bounds if the failure probability is not required to be polynomially
small.

\begin{lemma}[\cite{ah-adrp-08,fjops-endgt-16}]
	\lemlab{est:set:emptiness}%
	Consider a set $\SX \subseteq \SA$, where $n = \cardin{\SA}$.  Let
	$\eps \in (0,1)$ be a parameter. Assume that one is given an
	access to an emptiness oracle that, given a query set
	$\SQ \subseteq \SA$, returns whether or not
	$\SX \cap \SQ \neq \varnothing$.  Then, one can compute an
	estimate $s$ such that
	$(1-\eps)\cardin{\SX}\leq s \leq (1+\eps)\cardin{\SX}$, using
	$O( \eps^{-2} \log n)$ emptiness queries. The returned estimate is
	correct with probability $\geq 1 - 1/n^{\Omega(1)}$.
\end{lemma}

We sketch the basic idea of the algorithm used in the above lemma.
For a guess $g$ of the size of $\SX$, consider a random sample $\SQ$
where every element of $\SA$ is picked with probability $1/g$. The
probability that $\SQ$ avoids $\SX$ is
$\alpha(g) = (1-1/g)^{\cardin{\SX}}$. The function $\alpha(g)$ is:
\begin{enumerate*}[label=(\roman*)]
	\item monotonically increasing,
	\item close to zero when $g \ll \cardin{\SX }$,
	\item $\approx 1/e$ for $g=\cardin{\SX}$, and
	\item close to $1$ if $g \gg \cardin{\SX}$.
\end{enumerate*}
One can estimate the value $\alpha(g)$ by repeated random sampling and
checking if the random sample intersects $\SX$ using emptiness
queries. Given such an estimate one can then perform an approximate
binary search for the value of $g$ such that $\alpha(g) = 1/e$, which
corresponds to $g= \cardin{\SX}$.  See
\cite{ah-adrp-08,fjops-endgt-16} for further details.

\section{Edge sparsification by random coloring}
\seclab{sparsification}%

In this section, we present and prove that coloring vertices, and
counting only edges between specific color classes provides a reliable
estimate for the number of edges in the graph. This is distinct from
standard graph sparsification algorithms which usually sparsify the
edges of the graph directly (usually, by sampling edges).

We need the following technical lemma.
\begin{lemma}
	\lemlab{silly}%
	Let $C$ be a set of $r$ elements, colored randomly by $k$ colors
	-- specifically, for every element $x \in C$, one chooses randomly
	(independently and uniformly) a color for it from the set
	$\IRX{k}$. For $i \in \IRX{k}$, let $\ColX{i}$ be the number of
	elements of $C$ with color $i$.  Let $n$ be a positive integer and
	$c > 1$ be an arbitrary constant. Then:
	\medskip%
	\begin{compactenumA}
		\item For any color $i \in \IRX{k}$, we have
		\begin{math}
		\Prob{ \smash{ | \ColX{i} - r/k | > \sqrt{(c r/2) \ln n }}
			\,\bigr.} \leq 2/n^c.
		\end{math}
		
		\smallskip
		\item For any two distinct colors $i,j \in \IRX{k}$, we have
		\begin{math}
		\Prob{ \smash{ | \ColX{i} - \ColX{j} | > \sqrt{2 c r \ln n
			}}\, \bigr.} \leq 4/n^c.
		\end{math}
		
		\smallskip%
		\item For any two distinct colors $i,j \in \IRX{k}$, we have
		$\Ex{\bigl. \cardin{ \ColX{i} - \ColX{j} } } \leq \sqrt{
			2r/k}$.
	\end{compactenumA}
\end{lemma}
\begin{proof}
	(A) For $\ell\in\IRX{r}$, let $X_\ell\big.$ be the indicator
	variable that is $1$ with probability $1/k$ and $0$ otherwise.
	For $X = \sum_{\ell=1}^{r} X_\ell$, notice that $n_i$ is
	distributed identically to $X$, and that
	$\Ex{X} = \Ex{n_i} = r/k$. Using Chernoff's inequality
	(\lemref{chernoff} \itemref{Chernoff:A}), we have
	\begin{equation*}
		\Prob{ \cardin{ X - \Ex{X}} >
			\sqrt{(c r/2) \ln n }}%
		\leq%
		2 \exp\pth{  - \tfrac{2}{r}  \cdot  \tfrac{cr}{2} \ln n}%
		\leq%
		2/n^c,
	\end{equation*}

	(B) Observe that
	$\cardin{\ColX{i} -\ColX{j}} \leq \cardin{\ColX{i} - r/k} +
	\cardin{r/k -\ColX{j}}$, and the claim follows from (A).
	
	\smallskip%
	(C) For $t=1,\ldots,r$, let $X_t =1$ if the $t$\th element of $C$
	is colored by color $i$, and let $X_t = -1$ if this element is
	colored by color $j$. Otherwise, set $X_t=0$.  Clearly, the
	desired quantity is $\mu = \Ex{\bigl.|X|}$, where
	$X = \sum_{t=1}^r X_t\Bigr.$.  We have that
	$\Prob{X_t = 1 } =\Prob{X_t = -1} = 1/k$, $\Ex{X_t} = 0$, and that
	$\Ex{X_t^2} = 2/k$.  As such, by the independence of the $X_i$s,
	we have
	\begin{math}
	\Ex{X^2}%
	=%
	2\sum_{i<j} \Ex{ X_{i}X_j } + \sum_{i=1}^r \Ex{X_i^2} =%
	r \frac{2}{k}.\Bigr.
	\end{math}
	Finally, we have
	\begin{math}
	\Var{|X|}%
	=%
	\Ex{\bigl. |X|^2} - \mu^2 %
	\geq 0.
	\end{math}
	As such,
	$\mu = \Ex{\bigl.|X|} \leq \sqrt{ \Ex{\bigl.|X|^2}} = \sqrt{
		\Ex{\bigl.X^2}} \leq \sqrt{2r/k}$.
\end{proof}

\begin{lemma}
	\lemlab{sparsification_general}%
	(A) There exists an absolute constant $\varsigma$ such that the
	following holds. For every $n$, let $\Graph = (\IRX{n},\Edges)$ be
	a graph with $\nEdges$ edges. For any
	$1 \leq k \leq \lfloor n/2 \rfloor$, let $\SA_1,\dots,\SA_{2k}$ be
	a uniformly random partition of $\IRX{n}$. Then,
	\begin{equation*}
		\ProbChar \biggl[%
		\Bigl|\displaystyle \frac{\nEdges}{2k} - \displaystyle
		\sum_{i=1}^{k} \mY{\SA_i}{ \SA_{k+i}} \Bigr|%
		\geq \ctau \displaystyle \sqrt{\nEdges} \log n%
		\biggr] %
		\leq%
		\displaystyle \frac{1}{n^{4}}, %
		\quad\text{and }\quad
		\ProbChar{\biggl[%
			{\Bigl|\frac{\nEdges}{2k} - \sum_{i=1}^{2k} \mX{\SA_i} \Bigr|}
			\geq \ctau
			\sqrt{\nEdges}\log n
			\biggr]}
		\leq  \frac{1}{n^{4}}.
	\end{equation*}
	
	(B) There exists an absolute constant $\varsigma$ such that the
	following holds. Similarly, for every $n$, disjoint sets
	$\SA,\SB \subseteq \IRX{n}$ and $k$ such that
	$2\leq k\leq \max\{|\SA|, |\SB|\}$, let $\SA_1,\dots,\SA_k$,
	$\SB_1,\dots,\SB_k$ be uniformly random partitions of $\SA$ and
	$\SB$, respectively. Then,
	\begin{equation*}
		\ProbChar { \biggl[%
			\Bigl| \mY{\SA}{\SB} - k\sum_{i=1}^{k} \mY{\SA_i}{\SB_{i}}
			\Bigr|%
			\geq%
			\ctau k \sqrt{  \mY{\SA}{\SB} } \log n \biggr] }%
		\leq%
		\displaystyle {1}/{n^{4}}.
	\end{equation*}
\end{lemma}
\begin{proof}
	(A) Consider the random process that colors vertex $t$, at time
	$t \in \IRX{n}$, with a uniformly random color $Y_t \in
	\IRX{2k}$. The colors correspond to the partition of $\IRX{n}$
	into classes $\SA_1,\ldots, \SA_{2k}$. Define
	\begin{equation*}
		f(Y_1,\ldots,Y_n) = \sum_{i=1}^{k} \mY{\SA_i}{\SA_{k+i}}.        
	\end{equation*}
	The probability of a specific edge $uv$ to be counted by $f$ is
	$1/(2k)$. Indeed, fix the color of $u$, and observe that there is
	only one choice of the color of $v$, such that $uv$ would be
	counted. As such, $\Ex{f} = \nEdges/(2k)$ and
	$0\leq f(Y_1,\ldots,Y_n) \leq \nEdges$.
	
	Consider the Doob martingale $X_0,X_1, \ldots, X_n$, where
	$X_t = \ExCond{f(Y_{\IRX{n}})}{Y_{\IRX{t}}}$, where
	$Y_{\IRX{t}} \equiv Y_1, \ldots, Y_t$.  We are interested in
	bounding the quantity
	\begin{math}
	\cardin{X_{t} - X_{t-1}}.
	\end{math}
	To this end, fix the value of $Y_{\IRX{t-1}}$, and let
	\begin{equation*}
		g(\alpha ) = \ExCond{f(Y_{\IRX{n}})}{Y_{\IRX{t-1}} \cap (Y_t =
			\alpha)} .        
	\end{equation*}
	We have that
	\begin{math}
	X_{t-1} = \ExCond{f(Y_{\IRX{n}})}{Y_{\IRX{t-1}}} %
	=%
	\sum_{\alpha=1}^{2k} g(\alpha)/2k.
	\end{math}
	Namely, the value of $X_{t-1}$ is an average of the values in
	$\mathcal{G} = \brc{g(1), g(2), \ldots, g(2k)}$. Clearly,
	$X_t \in \mathcal{G}$. As such, we have that
	$\cardin{X_t - X_{t-1}} \leq \max_{i,j} |g(i) - g(j)|$.

	Let $N(t)$ be the set of neighbors of $t$ in the graph and
	$\degX{t}=|N(t)|$ be the degree of $t$. Let
	\begin{math}
	N_{<t} = N(t) \cap \IRX{t-1}
	\end{math}
	and
	\begin{math}
	N_{>t} = N(t) \cap \IRY{t+1}{n}
	\end{math}
	be the before/after set of neighbors of $t$, respectively. Let
	$\ColLTY{t}{i}$ (resp. $\ColGTY{t}{i}$) be the number of neighbors
	of $t$ in $N_{<t}$ (resp. $N_{>t}$) colored with color $i$.  For a
	color $i \in \IRX{2k}$, let $\pi(i) = 1 + ((k + i - 1) \bmod 2k)$
	be its matching color.
	
	Fix two distinct colors $i,j \in \IRX{2k}$, and let
	\begin{equation*}
		\Delta_t%
		= %
		\cardin{ \bigl. g(i) - g(j) }
		=%
		\cardin{%
			\ColLTY{t}{ \pi(i)}
			+ \ExCond{ \smash{\ColGTY{t}{ \pi(i)} } \bigr.\! }{Y_t=i}
			- \ColLTY{t}{\pi(j)}
			- \ExCond{ \smash{ \ColGTY{t}{\pi(j) }} \bigr.\! }{Y_t=j}
		}.
	\end{equation*}
	To see why the above is true, observe that any edge involving two
	vertices in $\IRX{t-1}$ has the same contribution to $g(i)$ and
	$g(j)$. Similarly, an edge with a vertex in $\IRX{t-1}$, and a
	vertex in $\IRY{t+1}{n}$, has the same contribution to both terms.
	The same argument holds for an edge involving vertices with
	indices strictly larger than $t$. As such, only the edges adjacent
	to $t$ have a different contribution, which is as stated.
	Rearranging, we have by \lemref{silly} with $C=N(t)$ and
	$r=\degX{t}$, with probability at least $1-\beta$ for
	$\beta = 4/n^c$ for any constant $c>1$, that
	\begin{equation*}
		\Delta_t%
		=%
		\cardin{%
			\ColLTY{t}{ \pi(i)}
			+ \Ex{ \smash{\ColGTY{t}{ \pi(i)} } \bigr.}
			- \ColLTY{t}{\pi(j)}
			- \Ex{ \smash{ \ColGTY{t}{\pi(j) }}  \bigr. }
		}
		\leq%
		\cardin{ \ColLTY{t}{ \pi(i)} - \ColLTY{t}{\pi(j)}}
		+ %
		\Ex{ \Bigl.\cardin{ \smash{\ColGTY{t}{ \pi(i)} }
				-  \smash{ \ColGTY{t}{\pi(j) }}  \bigr. }}
		\leq%
		c_t,
	\end{equation*}
	for
	\begin{math}
	c_t =%
	\sqrt{2c\;\degX{t} \ln n } + \sqrt{\degX{t}/k}%
	\leq%
	3c \sqrt{ \degX{t} \ln n }.
	\end{math}
	Let $\BadEvent$ be the event that any $\Delta_t$ (for any choice
	of $i$, $j$ or $t$) exceeds $c_t$, and observe that we can choose
	a constant $c>1$ such that
	$\Prob{\BadEvent} \leq (2k)^2 n \beta \leq 1/n^{10}$.
	
	Let
	\begin{math}
	S%
	=%
	\sum_{t=1}^n c_t^2 =%
	\sum_{t=1}^n 9 c^2 \degX{t} \ln n =%
	O( \nEdges \ln n),
	\end{math}
	and
	\begin{math}
	s%
	=%
	\ctau \sqrt{ \nEdges} \ln n,
	\end{math}
	Applying \lemref{azuma:cond} to $X_{\IRX{n}}$, we have
	\begin{equation*}
		\Prob{ \bigl. \cardin{f - {m}/{2k}} > s}
		=%
		\Prob{ \bigl. \cardin{X_n - X_0} > s}%
		\leq %
		2\exp\pth{ - s^2/ 2S } + \Prob{\BadEvent}
		\leq%
		2/n^{10} + 1/n^{10}%
		\leq%
		1/n^4,
	\end{equation*}
	for $\ctau$ sufficiently large.
	
	For the second claim in part (A), a nearly-identical argument
	works, with $f(Y_1,\ldots,Y_n) = \sum_{i=1}^{2k} \mX{\SA_i}$. Part
	(B) also follows by a similar argument as part (A),
	e.g. identifying $\SB_i$ with $\SA_{k+i}$ throughout.
\end{proof}

\begin{remark}%
	\remlab{threshold:coloring}%
	Given an induced bipartite graph $\Graph = (\SA, \SB, \Edges)$
	with $\nEdges$ edges, coloring it with $k$ colors, and taking the
	bipartite subgraphs of the resulting matching of the coloring, as
	done in \lemref{sparsification_general}, results in $k$ new
	disjoint bipartite (induced) subgraphs,
	$\Graph_i = (\SA_i, \SB_i, \Edges_i )$, for $i=1,\ldots, k$, with
	total number of edges $\Gamma = \sum_{i=1}^k \mX{\Graph_i}
	$. Furthermore, we have that $k \cdot \Gamma$ is a
	$(1\pm \epsA)$-approximation to $\mX{\Graph}$, where
	\begin{math}
	\epsA = \bigl( \ctau k \sqrt{ \nEdges } \log n \bigr) /
	\nEdges,
	\end{math}
	with high probability.  For our purposes, we need
	\begin{equation*}
		\epsA \leq \frac{\eps}{8 \log n}%
		\iff%
		\frac{\bigl( \ctau k \sqrt{ \nEdges } \log n \bigr)}%
		{\nEdges}
		\leq
		\frac{\eps}{8 \log n}
		\iff%
		\frac{8\bigl( \ctau k  \log^2 n \bigr)}{\eps}%
		\leq
		\sqrt{\nEdges}
		\iff%
		\nEdges = \Omega(  k^2 \eps^{-2}  \log^4 n ).
	\end{equation*}
	
	Setting $k=4$, the above implies that one can apply the refinement
	algorithm of \lemref{sparsification_general} if
	$\nEdges = \Omega(\eps^{-2} \log^4 n )$. With high probability,
	the number of edges in the new $k$ subgraphs (i.e., $\Gamma$),
	scaled by $k$, is a good estimate (i.e., within a
	$1 \pm \eps/(8 \log n)$ factor) for the number of edges in the
	original graph, and furthermore, the number of edges in the new
	subgraphs is small (formally, $\Ex{\Gamma} \leq \nEdges / 4$, and
	with high probability $\Gamma \leq \nEdges/2$).
\end{remark}

\section{Edge estimation using \BIS queries}
\seclab{bisq}

Here we show how to get exact and approximate count for the number of
edges in a graph using \BIS queries.

\subsection{Exactly counting edges using \BIS queries}
\seclab{bis:exact}

\begin{lemma}
	\lemlab{bisq:exact}%
	Given two disjoint sets $\SA,\SB \subseteq \IRX{n}$, one can
	(deterministically) compute $\EdgesY{\SA}{\SB}$, and thus
	$\mY{\SA}{\SB} = \cardin{\EdgesY{\SA}{\SB}}$, using
	$O( 1 + \mY{\SA}{\SB} \log n )$ \BIS queries.  Alternatively,
	given a parameter (informally, a query budget)
	$t = \Omega( \log n)$, one can decide if the given graph has
	$\leq t / \log n$ edges (or more) using $O(t)$ \BIS queries.
\end{lemma}

\begin{proof}
	We use a recursive divide-and-conquer approach, which intuitively
	builds a quadtree over the pair $(\SA,\SB)$.
	Specifically, consider the incidence matrix $M$ of size
	$\cardin{\SA}\times\cardin{\SB}$, where a column corresponds to
	an element of $\SB$, and a row to an element of $\SA$. An entry
	in the matrix is equal to one if there is an edge between the
	corresponding nodes in the original graph, and it is zero
	otherwise. The task at hand as such is to count the number of
	ones in the matrix. A \BIS query then corresponds to deciding
	if an induced submatrix is all zero. We now conceptually build
	a tree (i.e., a quadtree), by partitioning the matrix into four
	submatrices of the same dimensions (in the natural way), and
	recursively build a quadtree for each submatrix. Intuitively,
	the algorithm counts the $1$s in the matrix, by tracking each
	of the $1$s to their corresponding leaf node in the quadtree.

	To this end, the algorithm first issues the query
	\BIS{}$(\SA,\SB)$. If the return value is false, then there are no
	edges between $\SA$ and $\SB$, and the algorithm sets
	$\mY{\SA}{\SB}$ to zero, and returns. If $|\SA| = |\SB| =1$, then
	this also determines if $\mY{\SA}{\SB}$ is $0$ or $1$ in this
	case, and the algorithm returns. The remaining case, is that
	$\mY{\SA}{\SB} \neq 0$, and the algorithm recurses on the four
	children of $(\SA,\SB)$, which will correspond to the pairs
	$(\SA_1,\SB_1), (\SA_1,\SB_2),(\SA_2,\SB_1)$, and $(\SA_2,\SB_2)$,
	where $\SA_1,\SA_2$ and $\SB_1,\SB_2$ are equipartitions of $\SA$
	and $\SB$, respectively. We are using here the identity
	\begin{equation*}
		\mY{\SA}{\SB} 
		=%
		\mY{\SA_1}{\SB_1} + \mY{\SA_1}{\SB_2} + \mY{\SA_2}{\SB_1}
		+ \mY{\SA_2}{\SB_2}.
	\end{equation*}
	
	If $\mY{\SA}{\SB} = 0$ holds, then the number of queries is
	exactly equal to $1$, and the lemma is true in this case. For the
	rest of the proof we assume that $\mY{\SA}{\SB} \geq 1$.  To bound
	the number of queries, imagine building the whole quadtree for the
	adjacency matrix of $\SA \times \SB$ with entries for
	$\EdgesY{\SA}{\SB}$. Let $X$ be the set of 1 entries in this
	matrix, and let $k = |X|$ (i.e., $X$ corresponds to set of leaves
	that are labeled 1 in the quadtree). The height of the quadtree is
	$h = O( \max\{\log |\SA|, \log |\SB| \} )$. Let $X_1$ be the set
	of nodes in the quadtree that are either in $X$ or are ancestors
	of nodes of $X$. It is not hard to verify that
	\begin{math}
	\cardin{X_1} = O\bigl( k + k \log( |\SA| |\SB| ) \bigr)%
	=%
	O( k \log n).
	\end{math}
	Finally, let $X_2$ be the set of nodes in the quadtree that are
	either in $X_1$, or their parent is in $X_1$. Clearly, the
	algorithm visits only the nodes of $X_2$ in the recursion, thus
	implying the desired bound.
	
	As for the budgeted version, run the algorithm until it has
	accumulated $T = O(t/\log n)$ edges in the working set, where
	$T > t$. If this never happens, then the number of edges of the
	graph is at most $T$, as desired, and the above analysis
	applies. Otherwise, the algorithm stops, and applying the same
	argument as above, we get that the number of \BIS queries is
	bounded by $O( T \log n)=O(t)$.
\end{proof}

\begin{remark:unnumbered}
	The number of \BIS queries made by the algorithm of
	\lemref{bisq:exact} is at least $\max\{\mY{\SA}{\SB},1\}$, since
	every edge with one endpoint in $\SA$ and the other in $\SB$ is
	identified (on its own, explicitly) by such a query.
\end{remark:unnumbered}

Though we do not need it in sequel for our algorithms to estimate the
number of edges in a graph, we can use the above algorithm to exactly
identify the edges of an \emph{arbitrary} graph using \BIS queries
with a cost of $O(\log n)$ overhead per edge.

\begin{lemma}
	\lemlab{adjacent:edges}%
	Given a vertex $v \in \IRX{n}$, one can compute all the edges
	adjacent to $v$ in $\Graph$ using $O(1+\degX{v} \log n)$ queries.
\end{lemma}

\begin{proof}
	Let $\SB = \{v\}$, and $\SA = \IRX{n} \setminus \SB$, and observe
	that $\degX{v} = \mY{\SB}{\SA}$.  The algorithm of
	\lemref{bisq:exact} can now be used, observing that it can be
	modified to report all the edges found, thus implying the result.
\end{proof}

\begin{lemma}
	\lemlab{lem:bfs}%
	Given a vertex $v \in \IRX{n}$, and a graph
	$\Graph = (\IRX{n}, \Edges)$, let $\CC$ be the connected component
	of $v$ in $\Graph$. The set of edges in $\CC$ (i.e.,
	$\EdgesX{\CC}$) can be computed using $O( 1 + \mX{\CC} \log n )$
	\BIS queries, where $\mX{\CC}$ is the number of edges in $\CC$.
\end{lemma}

\begin{proof}
	Do a \BFS in $\Graph$ starting from $v$. Whenever reaching a
	vertex for the first time, compute its adjacent edges using
	\lemref{adjacent:edges}. Clearly, the \BFS visits all the vertices
	in $\CC$, and therefore computes all the edges in this connected
	component. The bound on the number of queries readily follows by
	observing that $\sum_{v \in \VX{\CC}} d(v) \log n$ is
	$O(\mX{\CC} \log n)$.
\end{proof}

\begin{lemma}
	\lemlab{bisq:exact:all}%
	For a graph $\Graph=(\IRX{n},\Edges)$, one can deterministically
	compute $\mX{\Edges}$ exactly, using at most
	$O( \log n + |\Edges| \log n )$ \BIS queries.  Alternatively,
	given a parameter (informally, a query budget)
	$t = \Omega( \log n)$, one can decide whether the given graph has
	at most $t / \log n$ edges, or more than this number, using $O(t)$
	\BIS queries.
\end{lemma}

\begin{proof}
	If we are given a set $Z$ that contains at least one vertex in
	each connected component of $\Graph$, then the result follows
	readily by applying the algorithm of \lemref{lem:bfs} to the
	vertices of $Z$ in order, deleting the vertices of each
	connected component as it is being discovered from $Z$. The
	total number of \BIS queries is $O( |\Edges| \log n )$, as the
	edges of each connected component are discovered in different
	invocations of \lemref{lem:bfs}.
	
	We remain with the task of computing $Z$. Let
	$\SA_0 = \IRX{n}$. For $i = 1, \ldots, T = \ceil{\log_2 n}$,
	let $A_i$ be the elements of $\SA_{i-1}$ whose $i$\th bit in
	their binary representation is $1$. Let
	$B_i = \SA_{i-1} \setminus A_i$.  Compute all the edges in
	$\Edges_i = \EdgesY{A_i}{B_i}$ using the algorithm of
	\lemref{bisq:exact}. This requires $O( 1+ |\Edges_i| \log n )$
	queries. Let $\VRT_i = \VX{\Edges_i}$. We add $\VRT_i$ to $Z$,
	and set $\SA_i = \SA_{i-1}\setminus \VRT_i$.
	
	Observe that every edge $e$ in $\Graph$ has an index $i$ such
	that its two vertices differ in the $i$\th bit. Note that either
	one of the endpoints of $e$ was already added to $Z$ before the
	$i$\th iteration, or it would be discovered and its endpoints
	added to $Z$ in the $i$\th iteration.  As such, the set $Z$ is
	computed correctly.  Since $\Edges_1, \ldots, \Edges_T$ are
	disjoint sets, it follows that computing $Z$ requires
	$O( T + \sum_i \cardin{\Edges_i}) = O( \log n + |\Edges| \log n
	)$ \BIS queries.
	
	\medskip
	
	For the budgeted version, we run the algorithm until
	$\tau = \Omega(t)$ \BIS queries have been performed. If this does not
	happen, then the graph has at most $\tau$ edges, and they were
	reported by the algorithm. Otherwise, we know that the graph must
	have at least $\tau/\log n$ edges, as desired.
\end{proof}

\subsection{The Coarse Estimator algorithm}
\seclab{coarse}

Let $\Graph = \Graph(\IRX{n}, \Edges)$ be a graph and let
$\SA, \SB \subseteq \IRX{n}$ be disjoint subsets of the vertices.  The
task at hand is to estimate $\mY{\SA}{\SB}$, using polylog $\BIS$
queries.

For a subset $S \subseteq \IRX{n}$, define $N(S)$ to be the union of
the neighbors of all the vertices in $S$. For a vertex $v$, let
$\deg_S(v)$ denote the number of neighbors of $v$ that lie in $S$. For
$i \in \IRX{\log n}$, define the set of vertices in $\SA$ with degree
between $2^i$ and $2^{i+1}$ as
\begin{equation*}
	\SA_i = \Set{\eA}{\eA \in \SA,\; 2^i < \deg_\SB(\eA) \leq
		2^{i+1}},    
\end{equation*}
and let $\SA_0$ denote the vertices in $\SA$ with
$\deg_\SB(v) \leq 2$.

\begin{claim}
	\clmlab{bucket}%
	There exists an $\idx \in \{0, 1, \ldots, \log n\}$ such that
	\begin{equation*}
		\mY{\SA_{\idx}}{\SB}%
		\geq%
		\frac{\mY{\SA}{\SB}}{\log n + 1}%
		\qquad \text{ and } \qquad%
		\cardin{\SA_{\idx}} \geq \frac{\mY{\SA}{\SB}}{2^{\idx+1} (\log
			n + 1)}.        
	\end{equation*}
\end{claim}
\begin{proof}
	Since $\sum_{i=0}^{\log n} \mY{\SA_i}{\SB}=\mY{\SA}{\SB}$, the
	first inequality is stating that there is a term as large as the
	average. As for the second inequality, observe that for every $i$,
	we have $|\SA_i|2^{i} \leq \mY{\SA_i}{\SB} \leq
	|\SA_i|2^{i+1}$. Hence, using the first inequality
	\begin{math}
	\ds |\SA_{\idx}|%
	\geq%
	\frac{\mY{\SA_{\idx}}{\SB}}{2^{\idx+1}} %
	\geq%
	\frac{\mY{\SA}{\SB}}{2^{\idx}}\cdot \frac1{2(\log n + 1)}.
	\end{math}
\end{proof}

Suppose that we have an estimate $\tm$ for the number of edges between
$\SA$ and $\SB$ in the graph. Consider the test \CheckEstimate,
depicted in \algref{correctness}, for checking if the estimate $\tm$
is correct up to polylogarithmic factors using a logarithmic number of
\BIS queries.

\begin{algorithm}[t]%
	\SetAlgoNoLine%
	\SetKwInOut{Query}{Query}%
	\KwIn{$((\SA,\SB), \tm)$ where $\SA,\SB \subseteq \IRX{n}$ are
		disjoint and $\tm$ is a (rough) guess for the value of
		$\mY{\SA}{\SB}$}%
	\smallskip%
	\hrule%
	\smallskip%
	\For{$i=0,1,\dotsc, \log n$}{%
		Sample $\SA' \subseteq \SA$ by choosing each vertex in $\SA$
		with probability $\min ( 2^i/ \tm, 1 )$.
		\\
		Sample $\SB' \subseteq \SB$ by choosing each vertex of $\SB$
		with probability ${1}/{2^i}$.\\
		\uIf{$\mY{\SA'}{\SB'} \neq 0$}{Output \textbf{accept}\;}%
	}%
	Output \textbf{reject}.
	\caption{\CheckEstimate{}$( \SA, \SB, \tm)$}
	\alglab{correctness}
\end{algorithm}

\needspace{3\baselineskip}
\begin{claim}
	\clmlab{checkest}%
	Let $n \geq 16$. If $\mY{\SA}{\SB} > 0$, then
	\begin{compactenumA}
		\smallskip
		\item \itemlab{check:A}%
		if $\tm \geq 4\mY{\SA}{\SB}(\log n+1)$, then
		\CheckEstimate{}$(\SA,\SB, \tm)$ accepts with probability at
		most ${1}/{4}$.
		
		\smallskip
		\item \itemlab{check:B}%
		if $\tm \leq \frac{\mY{\SA}{\SB}}{4\log n}$, then
		\CheckEstimate{}$(\SA,\SB, \tm)$ accepts with probability at
		least $1/2$.
	\end{compactenumA}
\end{claim}
\begin{proof}
	(A) For any value of the loop variable $i$, the probability that a
	fixed edge is present in the induced subgraph on $\SA'$ and $\SB'$
	is
	\begin{math}
	\min( 2^i/\tm, 1 ) \cdot ( 1/2^i ) \leq 1 / \tm.
	\end{math}
	Thus,
	\begin{math}
	\Ex{\mY{\SA'}{\SB'}}%
	\leq%
	{\mY{\SA}{\SB}}/{\tm} \leq \frac{1}{4(\log n+1)}.
	\end{math}
	For a fixed iteration $i$, by Markov's inequality, we have
	\begin{equation*}
		\Prob{\mY{\SA'}{\SB'} \neq 0}%
		= %
		\Prob{\mY{\SA'}{\SB'} \geq 1}%
		\leq%
		\Ex{\mY{\SA'}{\SB'}}%
		\leq%
		\frac{1}{4(\log n+1)}.        
	\end{equation*}
	By a union bound over the loop variable values, the probability
	that the test accepts is at most $1/4$.
	
	\medskip%
	\noindent%
	(B) It is enough to show that the probability is at least $1/2$
	when the loop variable attains the value $\idx$ given by
	\clmref{bucket}. In this case, we have that
	\begin{math}
	|\SA_\idx|%
	\geq%
	\frac{\mY{\SA}{\SB}}{2^{\idx+1} (\log n + 1)},
	\end{math}
	and thus
	\begin{align*}
		\Prob{\bigl. \SA' \cap \SA_\idx = \varnothing }%
		&= %
		\Bigl(1 - \frac{2^\idx}{\tm}\Bigr)^{|\SA_\idx|}
		\leq%
		\exp \pth{- \frac{2^\idx}{\tm} \cdot |\SA_\idx|}%
		\leq%
		\exp \pth{- \frac{2^\idx}{{\mY{\SA}{\SB}}/(4\log n)}
			\cdot
			\frac{\mY{\SA}{\SB}}{2^{\idx+1} (\log n + 1)} }%
		\\&%
		\leq%
		\exp \pth{- \frac{4\log n}{2 (\log n + 1)}}
		\leq%
		\frac{1}{e^{1.6}},
	\end{align*}
	since $n\geq 16$.  Furthermore, since $\deg_\SB(\eA) \geq 2^\idx$
	for all $\eA \in \SA_\idx$, it follows that when
	$\SA' \cap \SA_\idx \neq \varnothing$, then
	$|N(\SA'\cap \SA_\idx)|\geq 2^\idx$. So, we can bound
	\begin{equation*}
		\ProbCond{\Bigl.\SB' \cap N(\SA' \cap \SA_\idx) =
			\varnothing}%
		{\SA' \cap \SA_\idx \neq \varnothing}%
		\leq%
		\left(1 - \frac{1}{2^\idx} \right)^{2^\idx} \leq \frac1e.        
	\end{equation*}
	From the above, we get
	\begin{align*}
		\Prob{\mY{\SA'}{\SB'} \neq 0}%
		&=%
		\Prob{\Bigl. \SA' \cap \SA_\idx \neq \varnothing}
		\cdot
		\ProbCond{\Bigl. \SB' \cap N(\SA'\cap \SA_\idx) \neq
			\varnothing}{\SA'\cap \SA_\idx\neq \varnothing }%
		\\&%
		\geq%
		\pth{1 - \frac1{e^{1.6}}}\left(1 - \frac1e\right) \geq \frac12. 
		~
	\end{align*}
\end{proof}

Armed with the above test, we can easily estimate the number of edges
up to a $O(\log n)$ factor by doing a search, where we start with
$\tm = n^2$ and halve the number of edges each iteration. The
algorithm is depicted in \algref{coarse_est}.

\begin{algorithm}[t]
	\SetAlgoNoLine%
	\SetKwInOut{Query}{Query}
	\KwIn{$(\SA,\SB)$ where $\SA, \SB \subseteq \IRX{n}$ are
		disjoint}%
	\KwOut{An estimate $\tm$ for the number of edges $\mY{\SA}{\SB}$
		computed using \BIS queries } \Indp%
	\smallskip%
	\lIf{$\mY{\SA}{\SB} = 0$}{Output 0}%
	\For{$j=2\log n,\dotsc,0$}{%
		Run $t := 128\log n$ independent trials of
		\CheckEstimate{}$(\SA,\SB, 2^j)$.\\
		\If{at least ${3t}/8$ of them output \textbf{accept}}%
		{Output $2^j$\;} }
	\caption{\CoarseEstimator{}$(\SA,\SB)$}
	\alglab{coarse_est}
\end{algorithm}

\begin{claim}
	\clmlab{coarse}%
	For $n \geq 16$, \CoarseEstimator{}$(\SA,\SB)$ outputs
	$\tm \leq n^2$ satisfying
	\[\frac{\mY{\SA}{\SB}}{8\log n} \leq \tm \leq 8\mY{\SA}{\SB}\log n,\]
	with probability at least $1 - {4n^{-4}\log n}$. The number of
	\BIS queries made is $\ce\log^3 n$ for a constant $\ce$.
\end{claim}
\begin{proof}
	For any fixed value of the loop variable $j$ such that
	$2^j \geq 4\mY{\SA}{\SB}(\log n+1)$, the expected number of
	accepts is at most ${t}/{4}$ using \clmref{checkest}
	\itemref{check:A}, where $t = 128\log n$. The probability that we
	see at least
	\begin{math}
	3t/8 = t/4+ t/8
	\end{math}
	accepts is bounded by
	\begin{math}
	\exp\pth{ -2 (t/8)^2 /t}%
	=%
	\exp(-t/32)%
	\leq%
	n^{-4}
	\end{math}
	by Chernoff's inequality (\lemref{chernoff} \itemref{Chernoff:A}).
	Taking the union over all values of $j$, the probability that the
	algorithm returns $2^j$, when $2^j \geq 4\mY{\SA}{\SB}(\log n+1)$,
	is at most $2n^{-4}\log n$.
	
	On the other hand, when $2^j \leq {\mY{\SA}{\SB}}/(4 \log n)$, the
	expected number of accepts is at least $t / 2$, by
	\clmref{checkest} \itemref{check:B}, and so the probability that
	we see at least $3t/8 = t/2 - t/8$ accepts is at least
	$1 - \exp\pth{-2t/8^2} \geq 1 - n^{-4}$ by Chernoff's inequality
	(\lemref{chernoff} \itemref{Chernoff:A}). Hence, conditioned on
	the event that the algorithm has not already returned a bigger
	value of $j$, the probability that we accept for the unique $j$
	that satisfies
	\begin{math}
	\mY{\SA}{\SB} / 8%
	\leq%
	2^j \log n%
	<%
	\mY{\SA}{\SB} / 4,
	\end{math}
	is at least $1-n^{-4}$.
	
	Overall, by a union bound, the probability that the estimator
	outputs an estimate $\tm$ that does not satisfy
	\begin{math}
	(8\log n)^{-1} \leq \tm/ \mY{\SA}{\SB} \leq 8\log n
	\end{math}
	is at most $4n^{-4}\log n$. The number of \BIS queries is bounded
	by
	\begin{math}
	O\pth{ \log^3n},
	\end{math}
	since for each value of $j$ there are $t=128\log n$ trials of
	\CheckEstimate, each of which makes $\log n+1$ queries to the \BIS
	oracle.
\end{proof}

Summarizing the above, we get the following result

\begin{lemma}%
	\lemlab{coarse_estimate}%
	For $n \geq 16$, and arbitrary $\SA,\SB \subseteq \IRX{n}$ that
	are disjoint, the randomized algorithm
	\CoarseEstimator{}$(\SA,\SB)$ makes at most $\ce\log^3 n$ \BIS
	queries {\em(}for a constant $\ce${\em)} and outputs
	$\est{e} \leq n^2$ such that, with probability at least
	$1-4n^{-4}\log n$, we have
	\begin{math}
	\ds \pth{8\log n}^{-1} \leq {\est{e}}/{\mY{\SA}{\SB}} \leq
	8\log n.
	\end{math}
\end{lemma}

\subsection{The overall \BIS approximation algorithm}
\seclab{bisalgo}
Given a graph $\Graph = (\IRX{n}, \Edges)$, we describe here an
algorithm that makes ${\polylog(n)}/{\eps^4}$ \BIS queries to estimate
the number of edges in the graph within a factor of $(1\pm \eps)$. %

The algorithm for estimating the number of edges in the graph is going
to maintain a data-structure $\DS$ containing:
\begin{compactenumA}
	\item An accumulator $\acc$ - this is a counter that maintains an
	estimate of the number of edges already handled.
	\item A list of triples
	$(\SA_1, \SB_1, \wtp_1), \ldots, (\SA_u, \SB_u, \wtp_u)$ where
	$\SA_i,\SB_i \subseteq \IRX{n}$ and $\wtp_1 > 1$ is a non-negative
	weight.
\end{compactenumA}
The \emphi{estimate} based on $\DS$ of the number of edges in the
original graph $\Graph = (\IRX{n}, \Edges)$ is
\begin{equation*}
	\mX{\DS} = \acc + \sum_i  \wtp_i\cdot \mY{\SA_i}{\SB_i}.
\end{equation*}
The number of \emphi{active} edges in $\DS$ is
\begin{math}
\mAX{\DS} = \sum_i \mY{\SA_i}{\SB_i}.
\end{math}

\medskip

\subsubsection{Cleanup, refine, and reduce}
\seclab{crr}

The algorithm uses three subroutines: cleanup, refine, and reduce,
described next.

\smallskip%
\begin{compactenumA}
	\item %
	\itemlab{a:cleanup}%
	\Algorithm{Cleanup}: The cleanup stage removes from $\DS$ all
	induced subgraphs that have few edges, by explicitly counting
	their number of edges.  Let
	\begin{equation}
	\Lsmall = \Theta(\eps^{-2} \log^4 n)%
	\eqlab{L:1}%
	\end{equation}
	as specified by \remref{threshold:coloring}.  Given the
	data-structure $\DS$, the algorithm scans the list of triples
	$(\SA, \SB, \wtp) \in~\DS$. For each triple $(\SA, \SB, \wtp)$,
	using the algorithm of \lemref{bisq:exact}, it decides if
	$\mY{\SA}{\SB} \leq 2\Lsmall$. If so, the value of $\mY{\SA}{\SB}$
	was just computed, and it adds $\wtp \cdot \mY{\SA}{\SB}$ to
	$\acc$. Finally, it removes this triple from $\DS$.
	
	If $\DS$ has no triples in it, then the algorithm returns $\acc$
	as the desired approximation.
	
	\medskip%
	\item %
	\itemlab{a:refine}%
	\Algorithm{Refine}: %
	We are given the data-structure $\DS$, where the graph associated
	with every triple has at least $\Lsmall$ edges. The algorithm
	replaces every triple $(\SA, \SB, \wtp) \in \DS$ by the four
	induced subgraphs resulting from $4$-coloring the graph
	$\Graph(\SA,\SB)$, as described by
	\lemref{sparsification_general}(B) (see also
	\remref{threshold:coloring}). Specifically, the coloring results
	in the pairs $(\SA_i,\SB_i)$, for $i=1,2,3,4$. The triple
	$(\SA, \SB, \wtp)$ is replaced in $\DS$ by the triples
	$\{(\SA_1, \SB_1, 4\wtp), \ldots, (\SA_4, \SB_4, 4\wtp)\}$. This
	increases the number of triples in $\DS$ by a factor of four.
	
	\smallskip%
	\item %
	\itemlab{a:reduce} %
	\Algorithm{Reduce}: %
	If $\DS$ has more than $2\Llen$ triples, where
	$\Llen = O(\eps^{-2} \log^8 n)$ as specified by \remref{L:2}, then
	the algorithm reduces the number of triples.
	
	To this end, the algorithm first computes for each triple
	$(\SA, \SB, \wtp) \in \DS$, a coarse estimate $\tilde e$ of the
	number of edges in $\mY{\SA}{\SB}$, such that
	$ \mY{\SA}{\SB}/(8\log n) \leq \tilde e \leq \mY{\SA}{\SB} 8\log
	n$, by using \algref{coarse_est}. This requires $O( \log^3 n )$
	\BIS queries per triple.
	
	Next, the algorithm uses the summation reduction algorithm of
	\lemref{importance:alg} applied to the list of triples in $\DS$,
	with $\epsA = \eps/ (8 \log n)$.  This reduces the number of
	triples in $\DS$ to be at most $\Llen$, while introducing a
	multiplicative error of $(1 \pm \epsA)$.
\end{compactenumA}

\subsubsection{The algorithm in detail}

The algorithm input is the graph $\Graph=(\IRX{n}, \Edges)$, and a
parameter $\eps>0$. Let
\begin{math}
\RT = O\pth{ \eps^{-4} \log^{14} n }
\end{math}
be some parameter. The algorithm works as follows.
\begin{compactenumA}
	\item Check if $\Graph$ has at most $O(\RT / \log^2 n)$ edges,
	using the algorithm of \lemref{bisq:exact:all}, which requires
	$O(\RT)$ \BIS queries. If so, the algorithm returns the exact
	number of edges in $\Graph$, and stops.
	
	\item \itemlab{a:2}%
	Compute a random $2$-coloring of the vertices of the graph,
	creating two sets $\SA \cup \SB = \IRX{n}$, see
	\lemref{sparsification_general} (A).  We now create a
	data-structure as described above, with
	$\DS = [ \acc, (\SA, \SB, 2) ]$, where $\acc$ is initialized to
	value $0$.
	
	\item As long as $\DS$ contains some triple the algorithm does the
	following:
	\begin{compactenuma}
		\item Performs \Algorithm{Cleanup} on $\DS$, as described in
		\secref{crr} \itemref{a:cleanup}.
		
		\smallskip%
		\item Performs \Algorithm{Refine} on $\DS$, as described in
		\secref{crr} \itemref{a:refine}.
		
		\smallskip%
		\item Performs \Algorithm{Reduce} on $\DS$, as described in
		\secref{crr} \itemref{a:reduce}.

	\end{compactenuma}
	
	\smallskip%
	\item The algorithm now returns the value $\acc$ as the desired
	approximation.
\end{compactenumA}

\subsubsection{Analysis}

\paragraph{Number of iterations.}

Initially, the number of active edges is at most $\nEdges$. Every time
\Algorithm{Refine} is executed, this number reduces by a factor of 2
with high probability using \lemref{sparsification_general}(B) (in
expectation, the reduction is by a factor of 4). As such, after
\begin{math}
\ceil{\log \nEdges} \leq \ceil{ \log \binom{n}{2}} \leq 2\log n
\end{math}
iterations there are no active edges, and then the algorithm
terminates.

\paragraph{Number of \BIS queries.}
Clearly, because \Algorithm{Reduce} is used on $\DS$ in each
iteration, the algorithm maintains the invariant that the number of
triples in $\DS$ is at most $O(\Llen)$, where
$\Llen = O(\eps^{-2} \log^8n)$ as specified by \remref{L:2}.

The procedure \Algorithm{Cleanup}, applies the algorithm of
\lemref{bisq:exact}, to decide whether a triple in the list has at
least $2\Lsmall$ edges associated with it, or fewer edges, where
$\Lsmall = \Theta(\eps^{-2} \log^4 n)$ (see \Eqref{L:1} and
\remref{threshold:coloring}). This takes $O( \Lsmall \log n)$ \BIS
queries.  Overall the \Algorithm{Cleanup} step performs is
$O( \Lsmall \Llen \log n)$ queries in each iteration.  The procedure
\Algorithm{Refine} does not perform any \BIS queries.  The procedure
\Algorithm{Reduce}, performs $O(\Llen \log^3 n)$ \BIS queries in the
estimation stage.

As such, overall, the algorithm performs
$O( \Lsmall \Llen \log n) = O\pth{ \eps^{-2} \log^4 n \cdot \eps^{-2}
	\log^8n \cdot \log n } = O( \eps^{-4} \log^{13} n )$ \BIS queries
per iteration. There are $O( \log n)$ iterations, and as such, the
overall number of \BIS queries is $\RT = O( \eps^{-4} \log^{14} n )$,
which also bounds the number of \BIS queries in the first step of the
algorithm.

\paragraph{Approximation error.}

The initial $2$-coloring of the graph, in \itemref{a:2}, introduces a
$(1\pm \eps_0)$-multiplicative error, by
\lemref{sparsification_general}, where
\begin{equation*}
	\eps_0 = O( \sqrt{  1/\nEdges } \log n ) \ll \epsA
	=%
	\frac{\eps}{8 \log n}.
\end{equation*}
Inside each iteration, \Algorithm{Cleanup} introduces no error. By the
choice of parameters, \Algorithm{Refine} introduces a multiplicative
error that is at most $1\pm \epsA$; see
\remref{threshold:coloring}. Similarly, \Algorithm{Reduce} introduces
a multiplicative error bounded by $1\pm \epsA$; see \remref{L:2}. As
such, the multiplicative approximation of the algorithms lies in the
interval
\begin{equation*}
	[(1-\eps_0) (1-\epsA)^{2\log n}, (1+\eps_0) (1-\epsA)^{2\log n}]%
	\subseteq  [1-\eps, 1+\eps],
\end{equation*}
since $(1-\eps/(8\log n))^{1 + 2 \log n} \geq 1-\eps$ and
\begin{math}
(1+\eps/(8\log n))^{1 + 2 \log n} \leq 1+\eps
\end{math}
as easy calculations show.

\paragraph{Probability of success.}
Throughout this analysis, $c$ will be a constant that can be chosen to
be arbitrarily large. The algorithm may fail due to the following
reasons: (i) the random two-coloring in Step (B) gives an estimate
that is far from its expectation $-$ this probability is at most
$1/n^c$ using \lemref{sparsification_general}(A); (ii) the Refine step
fails $-$ the probability for the failure of each iteration is at most
$1/n^c$ using \lemref{sparsification_general}(B); (iii) the coarse
estimate in Reduce step fails $-$ the probability for the failure of
each iteration is at most $1/n^c$ using \clmref{coarse}; and lastly
(iv) the summation reduction in the Reduce step fails $-$ the
probability for the failure of each iteration is at most $1/n^c$ using
\lemref{importance:alg}. Overall, every step performed by the
algorithm had probability at most $1/n^c$ to fail. The algorithm
performs $O( \polylog(n) )$ steps with high probability, which implies
that the algorithm succeeds with probability at least
$1 - 1/n^{O(1)}$.

\subsubsection{The overall \BIS result}

\begin{theorem}
	\thmlab{bisq}%
	Let $\Graph = (\IRX{n}, \Edges)$ be an undirected graph. For a
	parameter $\eps \in (0,1)$, one can compute an estimate $\estm$
	for the number of edges in $\Graph$, such that
	\begin{math}
	(1- \eps)\mG \leq \estm \leq (1+\eps)\mG,
	\end{math}
	where $\mG$ is the number of edges of $\Graph$. The algorithm
	performs $O( \eps^{-4} \log^{14} n )$ \BIS queries and succeeds
	with probability $\geq 1- 1/n^{O(1)}$.
\end{theorem}

\subsection{Degree estimation using \BIS queries}
\seclab{d:est}%
We provide an auxiliary degree estimation result, connecting \BIS
queries to local queries (e.g., \cite{f-sirvu-06, gr-aapg-08}).
\begin{lemma}
	\lemlab{d:est}%
	Given a graph $\Graph = (\IRX{n}, \Edges)$, a parameter
	$\eps \in (0,1)$, and a vertex $v \in \IRX{n}$, one can
	$(1\pm \eps)$-approximate $\degX{v}$ in $\Graph$ using
	$O( \eps^{-2} \log n)$ \BIS queries. The approximation is correct
	with high probability.
\end{lemma}
\begin{proof}
	Let $N(v) = \Set{ i }{ vi \in \Edges}$ be the set of neighbors of
	$v$, and let
	\begin{math}
	\Edges_v = \Set{ vi }{ vi \in \Edges}
	\end{math}
	be the corresponding set of edges.  We have
	$\degX{v} = \cardin{N(v)} = \cardin{\Edges_v}$.  Given a set of
	edges
	\begin{math}
	\Edges_\SQ%
	\subseteq%
	\Edges_v= \Set{vi}{i \in \IRX{n}},
	\end{math}
	the corresponding set of vertices is
	$\SQ = \Set{i}{vi \in \Edges_\SQ}$. In particular,
	$\SQ \cap N(v) \neq \varnothing$ $\iff$
	$\Edges_\SQ \cap \Edges_v \neq \varnothing$.  Deciding if
	$\Edges_\SQ \cap \Edges_v \neq \varnothing$ is equivalent to
	deciding if any of the edges adjacent to $v$ is in $\Edges_\SQ$,
	and this is answered by the \BIS query for $(\brc{v}, \SQ)$.
	Namely, the \BIS oracle can function as an emptiness oracle for
	$N(v) \subseteq \IRX{n}$.  Now, using the algorithm of
	\lemref{est:set:emptiness} we can $(1 \pm \eps)$-approximation
	$\cardin{N(v)}$ using $O(\eps^{-2} \log n)$ queries, as claimed.
\end{proof}

\section{Edge estimation using \IS queries}
\seclab{is:e:e} This section describes and analyzes our \IS query
algorithm (\thmref{is:approx}). At the end, we also discuss
limitations of \IS queries, suggesting that \IS queries may indeed be
weaker than \BIS queries.

\subsection{Exactly counting edges using \IS queries}
\seclab{quadtree} We start with an exact edge counting algorithm for
\IS queries. At a high-level, we use \lemref{bisq:exact} after
efficiently computing a suitable decomposition of our graph.

\begin{lemma}
	\lemlab{bipartite:like}%
	Given disjoint sets of vertices $\SA,\SB \subseteq \IRX{n}$, such
	that both $\SA$ and $\SB$ are independent sets, one can compute
	the number of edges $\mX{\SA \cup \SB}$ using
	$O( \mX{\SA \cup \SB} \log n)$ \IS queries, assuming
	$\mY{\SA}{\SB} > 0$.
\end{lemma}
\begin{proof}
	Since $\SA$ and $\SB$ are disjoint and independent, we have that
	$\mX{\SA\cup \SB} = \mY{\SA}{\SB}$. Furthermore, for any
	$\SA' \subseteq \SA$ and $\SB' \subseteq \SB$, the query
	$\BIS(\SA',\SB')$ is equivalent to the query
	$\IS(\SA' \cup \SB')$. As such, we can use the algorithm of
	\lemref{bisq:exact}, using the \IS queries as a replacement for
	the \BIS queries, yielding the result.
\end{proof}

The next step is to break the set of interest $\SA$ into independent
sets.

\begin{lemma}
	\lemlab{decompose}%
	Given a set $\SA \subseteq \IRX{n}$, one can decompose it into
	disjoint independent sets $\SB_1, \SB_2, \ldots, \SB_t$, such that
	\begin{compactenuma}
		\item $\SA = \bigcup_{i=1}^t \SB_i$, and
		\item for any $i,j \in \IRX{t}$, with $i < j$, we have
		$\mY{ \SB_i}{\SB_j} > 0$.
	\end{compactenuma}
	Furthermore, computing this decomposition uses only
	$O( 1 + \mX{\SA} \log n )$ \IS queries.
\end{lemma}

\begin{proof}
	Order the elements of $\SA = \brc{ \eA_1, \ldots, \eA_k} $ arbitrarily. The idea is to break $\SA$ into independent
	sets, where each independent set is an interval
	$I_j = \{ \eA_{i_j}, \eA_{i_j+1}, \ldots, \eA_{i_{j+1}-1}
	\}$. This can be done in a greedy fashion from left to right,
	discovering the index where an interval stops being an
	independent set.  Assume inductively that one has computed the
	first $j$ such independent intervals $I_1, \ldots, I_j$, and also assume that
	$I_j \cup \{ \eA_{i_{j+1}} \}$ is not an independent set.
	Next, using binary search on the range
	$\{i_{j+1}+1, \ldots, n\}$, find the maximal $\beta$ such that
	$\{\eA_{i_{j+1}}, \ldots, \eA_\beta\}$ is independent. Set
	$i_{j+2} = \beta+1$,
	$I_{j+1} = \{ \eA_{i_{j+1}}, \ldots, \eA_{i_{j+2}-1} \}$, and
	continue to the next iteration.  Note that each binary search for
	computing an interval uses $O( \log n)$ \IS queries.
	
	For any $j$, we have $\mY{I_j}{I_{j+1}} \geq 1$, which implies
	that the number of computed intervals $\tau$ satisfies $\tau \leq \mX{\SA} +1$.
	As such, this stage uses $O \bigl( (1+\mX{\SA}) \log n\bigr)$
	\IS queries.  This results in a decomposition of $\SA$ into $\tau$
	independent sets $I_1, \ldots, I_\tau$.
	
	\medskip%
	
	In the second stage, starting with the computed collection of
	independent sets, the algorithm greedily tries to merge
	sets. In each step, the algorithm takes two
	independent sets $\SC, \SD$ in the current collection (for
	which it might be possible that their merged set is
	independent), and the algorithm uses an \IS query to check whether
	$\SC \cup \SD$ is an independent set. If it is, then the
	algorithm merges the two sets into one independent
	set (replacing $\SC,\SD$ by the set $\SC \cup \SD$ in the
	current collection of sets). Otherwise, the algorithm marks the
	two sets $\SC$ and $\SD$ as being incompatible with each
	other. Note that if $\SC, \SD$ are incompatible, then for any
	$\SC' \supseteq \SC$ and $\SD' \supseteq \SD$, the sets $\SC'$
	and $\SD'$ are also incompatible. Namely, incompatibility is
	preserved under merger of independent sets, and the algorithm
	can keep track of the incompatible pairs under merger
	(importantly, a merger can not decrease the number of
	incompatible pairs). The algorithm stops when all the current
	sets are pairwise incompatible.
	
	Each merge of two independent sets can be charged to the number
	of independent sets decreasing by one. Each pair of sets that
	is discovered to be incompatible can be charged to the edge
	witnessing that the merged set is \emph{not}
	independent. Since every edge is only charged once by this
	process, it follows that the total number of \IS queries
	performed by the second stage of the algorithm is at most
	$\tau + \mX{\SA} \leq 2 \mX{\SA} + 1$.
	
	The resulting collection of independent sets has the
	desired properties, completing the proof.
\end{proof}

\begin{lemma}
	\lemlab{isq:exact}%
	Given $\SA \subseteq \IRX{n}$, one can deterministically compute
	$\EdgesX{\SA}$, using $O( 1 + \mX{\SA} \log n)$ \IS queries.
	Alternatively, given a budget $t> 0$ and set
	$\SA \subseteq \IRX{n}$, one can decide if $\mX{\SA} > t$ using
	$O( t \log n )$ \IS queries.
\end{lemma}
\begin{proof}
	Using the algorithm of \lemref{decompose}, compute the
	decomposition of $\SA$ into independent sets
	$\SB_1,\ldots, \SB_t$.  By construction, for any $i<j$, we have
	that $\mY{\SB_i}{\SB_j} \geq 1$, as some vertex of $\SB_i$ is
	connected to some vertex in $\SB_j$. As such, going over all
	$1\leq i<j\leq t$, compute the set of edges
	$\EdgesY{\SB_i}{\SB_j}$ using the algorithm of
	\lemref{bipartite:like}.  This requires
	$O(\mY{\SB_i}{\SB_j} \log n)$ \IS queries. As such, the total
	number of \IS queries used by this algorithm is
	\begin{math}
	O\bigl( \mX{\SA} \log n + \sum_{i < j} \mY{\SB_i}{\SB_j} \log
	n \bigr)%
	=%
	O( \mX{\SA} \log n).
	\end{math}

	\smallskip%
	
	The budgeted version follows by running the algorithm until
	$\ceil{ c \log n}$ \IS queries have been performed, for $c$ a
	sufficiently large constant. If this happens, then the number
	of edges in the graph is larger than $t$ (as otherwise the
	above implies that the algorithm would have already
	terminated), and the algorithm stops and outputs this fact. %
\end{proof}

\subsection{Algorithms for edge estimation using \IS queries}
\seclab{isq} Our \IS algorithm has two main subroutines. We first
describe and analyze these, then we combine them for the overall
algorithm, which is presented in \thmref{is:approx}.
\subsubsection{Growing Search}

The following is an immediate consequence of \lemref{isq:exact}.

\begin{lemma}
	\lemlab{g:s:small}%
	Let
	\begin{equation*}
		\Lbase = \ceil{\cA \eps^{-4} \log^4 n},
	\end{equation*}
	where $\cA$ is some sufficiently large constant.  Given a set
	$\SA$, one can decide if $\mX{\SA} \leq \Lbase$, and if so get the
	exact value of $\mX{\SA}$, using $O( \eps^{-4} \log^5 n )$ \IS
	queries.
\end{lemma}

\begin{lemma}
	\lemlab{next:size}%
	Given parameters $t$, $\eps \in (0,1]$, and a set
	$\SA \subseteq \IRX{n}$, such that
	$\mX{\SA} \geq \max(\Lbase,t^2)$, an algorithm can decide if
	$\mX{\SA} > 2t^2$, or alternatively return a
	$(1\pm \eps)$-approximation to $\mX{\SA}$ if
	$t^2 \leq \mX{\SA} \leq 2t^2$.  The algorithm uses
	$O( \eps^{-1} t \log^2 n)$ \IS queries and succeeds with
	probability $1 - 1/n^{O(1)}$.
\end{lemma}

\begin{proof}
	We color the vertices in $\SA$ randomly using
	$\ncols = \lceil t {\eps} / (\ctau \log n)\rceil$ colors for a
	constant $\ctau$ to be specified shortly, and let
	$\SA_1, \ldots, \SA_\ncols$ be the resulting partition.  By
	\lemref{sparsification_general}, we have for the estimate
	\begin{math}
	\Gamma = \sum_{i=1}^k \mX{\SA_i}
	\end{math}
	that
	\begin{equation*}
		\cardin{\mX{\SA} - \ncols\cdot \Gamma}
		\leq%
		\ctau \ncols \sqrt{ \mX{\SA}} \log n,
	\end{equation*}
	and this holds with probability $\geq 1-n^{-\cC}$, where $\cC$ is
	an arbitrarily large constant, and $\ctau$ is a constant that
	depends only on $\cC$. For this to be a
	$(1\pm \eps)$-approximation, we need that
	\begin{equation*}
		\frac{\ctau k \sqrt{ \mX{\SA}} \log n}{\mX{\SA}}
		\leq%
		\eps.%
	\end{equation*}
	This in turn is equivalent to
	\begin{equation*}
		\mX{\SA} \geq \pth{\frac{\ctau k  \log n}{\eps}}^2
		=%
		t^2,
	\end{equation*}
	which holds because of the assumption that
	$\mX{\SA} \geq \max\{\Lbase,t^2\}$ in the statement.
	
	To proceed, the algorithm starts computing the terms in the
	summation defining $\Gamma$, using the algorithm of
	\lemref{isq:exact}. If at any point in time, the summation exceeds
	$M = 8(t^2/k) = O( \eps^{-1} t \log n) $, then the algorithm stops
	and reports that $\mX{\SA} > 2t^2$. Otherwise, the algorithm
	returns the computed count $\ncols\cdot \Gamma$ as the desired
	approximation. In both cases we are correct with high probability
	by \lemref{sparsification_general}.
	
	We now bound the number of \IS queries. If the algorithm computed
	$\Gamma$ by determining exact edge counts for $m(U_i)$ for all
	$i \in \IRX{k}$, then the number of queries would be
	\begin{math}
	\sum_{i=1}^k O\pth{ 1 + \mX{\SA_i} \log n }.
	\end{math}
	However, the choice of stopping early if the number of queries
	exceeds $M= O( \eps^{-1} t \log n)$ implies that the total number
	of queries is bounded by
	$O( k + M \log n) = O \pth{ \eps^{-1} t \log^2 n}$.
\end{proof}

\begin{lemma}
	Given $\eps \in (0,1]$, and a set $\SA \subseteq \IRX{n}$, one can
	compute a $(1\pm \eps)$-approximation for $\mX{\SA}$. The
	algorithm uses at most
	$O( \eps^{-4} \log^5 n + \eps^{-1}\sqrt{\mX{\SA}} \log^2 n)$ \IS
	queries and succeeds with probability $1-1/n^{O(1)}$.
\end{lemma}

\begin{proof}
	The algorithm starts by checking if the number of edges in
	$\mX{\SA}$ is at most $\Lbase = O(\eps^{-4} \log^4 n)$ using the
	algorithm of \lemref{g:s:small}. Otherwise, in the $i$\th
	iteration, the algorithm sets $t_i = \sqrt{2} t_{i-1}$, where
	$t_0 = \sqrt{\Lbase}$, and invokes the algorithm of
	\lemref{next:size} for $t_i$ as the threshold parameter.  If the
	algorithm succeeds in approximating the right size we are
	done. Otherwise, we continue to the next iteration. Taking a union
	bound over the iterations, we have that the algorithm stops with
	high probability before $t_\alpha > 4 \sqrt{ \mX{\SA}}$.  Let
	$\alpha$ be the minimum value for which this holds. The number of
	\IS queries performed by the algorithm is
	\begin{math}
	O( \sum_{i=1}^\alpha t_i \eps^{-1} \log^2 n)%
	=%
	O( \eps^{-1} \sqrt{\mX{\SA}} \log^2 n),
	\end{math}
	since this is a geometric sum.
\end{proof}

\subsubsection{Shrinking Search}
\seclab{edge-existence}%

We are given a graph $\Graph = (\IRX{n}, \Edges)$, and a set
$\SA \subseteq \IRX{n}$. The task at hand is to approximate
$\mX{\SA}$. Let $\nSA = \cardin{\SA}$.

Given an oracle that can answer \IS queries, we can decide if a
specific edge $uv$ exists in the set $\EdgesX{\SA}$, by performing an
\IS query on $\brc{u,v}$. We can treat such \IS queries as membership
oracle queries in the set $\Edges$ of edges in the graph, where the
ground set is the set of all possible edges
$\AllEdges = \binom{\SA }{2} = \Set{ ij }{ i <j \text{ and } i,j \in
	\SA }$, where $\cardin{\AllEdges} = \nSA(\nSA-1)/2$. Invoking the
algorithm of \lemref{est:set} in this case, with
$\BadProb = 1/n^{O(1)}$, implies a $(1\pm \eps)$-approximation to
$\mX{\SA}$ using $O( (\nSA^2 / \mX{\SA}) \eps^{-2} \log n)$ \IS
queries. For our purposes, however, we need a budgeted version of
this.

\begin{lemma}
	\lemlab{is:prev:size}%
	Given parameters $t > 0 $, $\epsA \in (0,1]$, and a set
	$\SA \subseteq \IRX{n}$, with $\nSA = \cardin{\SA}$, an algorithm
	can return either: (a) $\mX{\SA} \leq \nSA^2 / (2t)$, or (b)
	return $(1\pm \epsA)$-approximation to $\mX{\SA}$.  The algorithm
	uses $O( t \log n)$ \IS queries in case (a), and
	$O( t \epsA^{-2} \log n)$ in case (b). The returned result is
	correct with high probability.
\end{lemma}
\begin{proof}
	The idea is to use the sampling as done in \lemref{estimate}, with
	$g= \nSA^2/(16t)$ and $\eps = 1/2$ on the sets of edges
	$\EdgesX{\SA} \subseteq \binom{\SA}{2}$. The sample $\Sample$ used
	is of size $O( (\nSA^2/ g) \log n) = O(t \log n)$, and we check
	for each one of the sampled edges if it is in the graph by using
	an \IS query. If the returned estimate is at most $g/2$, then the
	algorithm returns that it is in case (a).
	
	Otherwise, we invoke the algorithm of \lemref{estimate} again,
	with $\eps = \epsA$, to get the desired approximation, which is
	case (b).
\end{proof}

\subsubsection{The overall \IS Search algorithm}

\begin{theorem}
	\thmlab{is:approx}%
	We are given a graph $\Graph = (\IRX{n}, \Edges)$, with access to
	the edges of the graph via an \IS oracle.  Let
	$\nEdges = \cardin{\Edges}$ be the number of edges in
	$\Graph$. The quantity $\nEdges$ can be $(1\pm \eps)$-approximated
	by an algorithm that uses
	\begin{equation*}
		O\bigl( \eps^{-4} \log^5 n + \min( \sqrt{\nEdges}, n^2 / \nEdges )
		\eps^{-2} \log^2 n \bigr)
	\end{equation*}
	\IS queries, and it succeeds with probability $\geq 1-1/n^{O(1)}$.
\end{theorem}

\begin{proof}
	Let $t_0 = \ceil{\cA \eps^{-4} \log^4 n}$, for some constant
	$\cA$.  Using the algorithm of \lemref{g:s:small}, we can decide
	if $\nEdges \leq t_0$, and if so, return (the just computed)
	$\nEdges$.
	
	The algorithm now loops for $i=1,2,3,\ldots$, n. In the $i$\th
	iteration, it does the following:
	\begin{compactenumA}
		\smallskip%
		\item If $i=1$ then let $t_1 = \sqrt{t_0}$, otherwise set
		$t_i = 2t_{i-1}$.
		
		\smallskip%
		\item Using the algorithm of \lemref{next:size} decide if
		$\nEdges \leq 2t_i^2$, and if so it returns the desired
		$(1\pm \eps)$-approximation to $\nEdges$. This uses
		$O( t_i \eps^{-1} \log^2 n)$ \IS queries.
		
		\smallskip%
		\item Using the algorithm of \lemref{is:prev:size}, decide if
		$ \nEdges \leq n^2 / (2t_i)$, and if so continue to the next
		iteration. This uses $O( t_i \log n)$ \IS queries.
		
		Otherwise, the algorithm of \lemref{is:prev:size} returned the
		desired $(1\pm\eps)$-approximation, using
		$O( t_i \eps^{-2} \log n )$ \IS queries.
	\end{compactenumA}
	\medskip%
	Combining the two bounds on the \IS queries, we get that the
	$i$\th iteration used $O( t_i \eps^{-2} \log^2 n )$ \IS queries.
	
	The algorithm stopped in the $i$\th iteration, if
	$t_i \geq \sqrt{\nEdges/2}$, or $t_i \geq n^2 / \nEdges$. In
	particular, for the stopping iteration $I$, we have
	$t_I = O( \min( \sqrt{\nEdges}, n^2 / \nEdges ))$.  As such, the
	total number of \IS queries in all iterations except the last one
	is bounded by
	\begin{math}
	O( \sum_{i=1}^{I} t_i\eps^{-2} \log^2 n)%
	=%
	O( t_I \eps^{-2} \log^2 n).
	\end{math}
	The stopping iteration uses $O(t_I \eps^{-2} \log^2 n)$ \IS
	queries. Each bound holds with high probability, and a union bound
	implies the same for the final result.
\end{proof}

\begin{corollary}
	For a graph $\Graph = (\IRX{n}, \Edges)$, with an access to
	$\Graph$ via \IS queries, and a parameter $\eps > 0$, one can
	$(1\pm \eps)$-approximate $\nEdges$ using
	\begin{math}
	O( \eps^{-4} \log^5 n + n^{2/3} \eps^{-2} \log^2 n )
	\end{math}
	\IS queries.
\end{corollary}
\begin{proof}
	Follows readily as
	$\min( \sqrt{\nEdges}, n^2 / \nEdges ) \leq n^{2/3}$, for any
	value of $\nEdges$ between $0$ and $n^2$.
\end{proof}

\subsection{Limitations of \IS queries}
\seclab{limit} \seclab{is-degree}

In this section, we discuss several ways in which \IS queries seem
more restricted than \BIS queries.

\paragraph*{Simulating degree queries with \IS queries.}

A degree query can be simulated by $O(\log n)$ \BIS queries, see
\lemref{d:est}.  In contrast, here we provide a graph instance where
$\Omega\pth{{n} / {\degX{v}}}$ \IS queries are needed to simulate a
degree query. In particular, we show that \IS queries may be no better
than edge existence queries for the task of degree estimation. Since
it is easy to see that $\Omega\pth{ n / \degX{v} }$ edge existence
queries are needed to estimate $\degX{v}$, this lower bound also
applies to \IS queries.

For the lower bound instance, consider a graph which is a clique along
with a separate vertex~$v$ whose neighbors are a subset of the
clique. We claim that \IS queries involving $v$ are essentially
equivalent to edge existence queries. Any edge existence query can be
simulated by an \IS query. On the other hand, any \IS query on the
union of $v$ and at least two clique vertices will always detect a
clique edge. Thus, the only informative \IS queries involve exactly
two vertices.

\paragraph*{Coarse estimator with \IS queries.} 
It is natural to wonder if it is possible to replace the coarse
estimator (\lemref{coarse_estimate}) with an analogous algorithm that
makes $\polylog(n)$ \IS queries. This would immediately imply an
algorithm making $\polylog(n)/\e^4$ \IS queries that estimates the
number of edges. We do not know if this is possible, but one barrier
is a graph consisting of a clique $\SA$ on $O(\sqrt{\nEdges})$
vertices along with a set $\SB$ of $n - O(\sqrt{\nEdges})$ isolated
vertices. We claim that for this graph, the algorithm
\CoarseEstimator{}$(\SA,\SB)$ from \secref{coarse}, using \IS queries
instead of \BIS queries, will output an estimate $\wt \nEdges$ that
differs from $\nEdges$ by a factor of $\Theta(n^{1/3})$. Consider the
execution of \CheckEstimate{}$(\SA, \SB, \wtE)$ from
\algref{correctness}. A natural way to simulate this with \IS queries
would be to use an \IS query on $\SA' \cup \SB'$ instead of a \BIS
query on $(\SA',\SB')$. Assume for the sake of argument that
$\nEdges = n^{4/3}$ and $|\SA| = \sqrt{\nEdges} = n^{2/3}$. Consider
when the estimate $\wtE$ satisfies $\wtE = c n^{5/3}$ for a small
constant $c$. In the \CheckEstimate execution, there will be a value
$i = \Theta( \log n )$ such that, with constant probability,
$\SA' \subseteq \SA$ will contain at least two vertices and
$\SB' \subseteq \SB$ will contain at least one vertex. In this case,
$\mX{\SA'\cup \SB'} \neq 0$ even though $\mY{\SA'}{\SB'} = 0$. Thus,
using \IS queries will lead to incorrectly accepting on such a sample,
and this would lead to the \CoarseEstimator outputting the estimate
$\wtE = \Theta(n^{5/3})$ even though the true number of edges is
$\nEdges = n^{4/3}$.

\section{Conclusions}
\seclab{conclusion}

In this paper, we explored the task of using either \BIS or \IS
queries to estimate the number of edges in a graph. We presented
randomized algorithms giving a $(1+\eps)$-approximation using
$\polylog(n)/\e^4$ \BIS queries and
\begin{math}
\min\left \{n^2/(\e^2 \nEdges), \sqrt{\nEdges}/\e \right
\}\pcdot\polylog(n)
\end{math}
\IS queries. Our algorithms estimate the number of edges by first
sparsifying the original graph and then exactly counting edges
spanning certain bipartite subgraphs. Below we describe a few open
directions for future research.

\subsection{Open directions}

Open questions include using a polylogarithmic number of \BIS queries
to estimate the number of cliques in a graph (see~\cite{ers-ankcs-17}
for an algorithm using degree, neighbor and edge existence queries) or
to sample a uniformly random edge (see~\cite{er-seau-18} for an
algorithm using degree, neighbor and edge existence queries).  In
general, any graph estimation problems may benefit from \BIS or \IS
queries, possibly in combination with standard queries (such as
neighbor queries). Finally, it would be interesting to know what other
oracles, besides subset queries, enable estimating graph parameters
with a polylogarithmic number of queries.

\paragraph{Acknowledgments.}
	We thank the anonymous referees for helpful comments about improving
	the presentation of our paper and for pointing out relevant
	references.
	
	Paul Beame was supported in part by NSF grant CCF-1524246. Sariel Har-Peled was supported in part by NSF AF awards CCF-1421231 and CCF-1217462 and this work was done while visiting the University of Washington on a sabbatical in 2017. 
	Sivaramakrishnan Natarajan Ramamoorthy was supported by the NSF under agreements CCF-1149637, CCF-1420268, CCF-1524251. This work was partially completed while Cyrus Rashtchian was a graduate student at the Paul G. Allen School of CSE, University of Washington, Seattle and was at Microsoft Research, Redmond. During the course of this work, Makrand Sinha was a graduate student at the Paul G. Allen School of CSE, University of Washington, Seattle and was supported by the NSF under agreements CCF-1149637, CCF-1420268, CCF-1524251.

\newcommand{\etalchar}[1]{$^{#1}$}
 \providecommand{\CNFX}[1]{ {\em{\textrm{(#1)}}}}
  \providecommand{\CNFISAAC}{\CNFX{ISAAC}}


\end{document}